\documentclass{article}

\usepackage{arxiv}
\usepackage{graphicx}
\usepackage[utf8]{inputenc} 
\usepackage[T1]{fontenc}    
\usepackage{hyperref}       
\usepackage{url}            
\usepackage{booktabs}       
\usepackage{amsfonts}       
\usepackage{nicefrac}       
\usepackage{microtype}      
\usepackage{lipsum}
\usepackage{threeparttable}
\usepackage{weiwAlgorithm}
\usepackage{url}
\usepackage{url,amsmath,amsthm,mathrsfs,amssymb,multicol,multirow}
\usepackage{tabularx,booktabs,mathrsfs,multicol,multirow}
\usepackage{adjustbox,subfigure}
\usepackage{wrapfig}

\newtheorem{example}{Example}

\newtheorem{theorem}{Theorem}
\newtheorem{lemma}{Lemma}

\newtheorem{definition}{Definition}

\newcommand{\ours}{BSRBK}

\newcommand{\myparagraph}[1]{\vspace{1mm} \noindent \textbf{#1}.}

\title{Efficient Top-k Vulnerable Nodes Detection in Uncertain Graphs}

\author{
Dawei Cheng \\
   \And
Chen Chen \\
   \And
Xiaoyang Wang \\
   \And
Sheng Xiang
}

\begin{document}
\maketitle

\begin{abstract}
Uncertain graphs have been widely used to model complex linked data in many real-world applications, such as guaranteed-loan networks and power grids, where a node or edge may be associated with a probability. In these networks, a node usually has a certain chance of default or breakdown due to self-factors or the influence from upstream nodes. For regulatory authorities and companies, it is critical to efficiently identify the vulnerable nodes, i.e., nodes with high default risks, such that they could pay more attention to these nodes for the purpose of risk management. In this paper, we propose and investigate the problem of top-$k$ vulnerable nodes detection in uncertain graphs. We formally define the problem and prove its hardness. To identify the $k$ most vulnerable nodes, a sampling-based approach is proposed. Rigorous theoretical analysis is conducted to bound the quality of returned results. Novel optimization techniques and a bottom-$k$ sketch based approach are further developed in order to scale for large networks. In the experiments, we demonstrate the performance of proposed techniques on 3 real financial networks and 5 benchmark networks. The evaluation results show that the proposed methods can achieve up to 2 orders of magnitudes speedup compared with the baseline approach. Moreover, to further verify the advantages of our model in real-life scenarios, we integrate the proposed techniques with our current loan risk control system, which is deployed in the collaborated bank, for more evaluation. Particularly, we show that our proposed new model has superior performance on real-life guaranteed-loan network data, which can better predict the default risks of enterprises compared to the state-of-the-art techniques.

\end{abstract}


\section{Introduction}
\label{sec:introduction}

Uncertainty is inherent in real-world data because of various reasons, such as the accuracy issue of devices and models~\cite{DBLP:conf/icde/LiDWMQY19,DBLP:journals/tkde/ZhangLZZZ16}. Sometimes, people may inject uncertainty to the data on purpose to protect user privacy~\cite{DBLP:journals/pvldb/BoldiBGT12}.
As a common data structure, graphs are widely used to model the complex relationships between different entities.
Due to the ubiquitous uncertainty, uncertain graph analysis has attracted significant attentions in the community of database management. A large number of graph problems have been studied in the context of uncertain graphs, e.g., nearest neighbor search~\cite{potamias2010k}, reliability query~\cite{ke2019depth}, cohesive subgraph mining~\cite{DBLP:conf/icde/LiDWMQY19}, etc.

In some real-life graphs, such as power grids and guaranteed-loan networks, nodes (e.g., facilities and enterprises) may breakdown or default due to self-factors or the issues from upstream nodes. To identify these high-risky (i.e., vulnerable) nodes, in this paper, we investigate a novel problem, named top-$k$ vulnerable nodes detection. Given a directed uncertain graph $\mathcal{G}$, each node $A$ (resp. edge $(A,B)$) is associated with a probability $p_s(A)$ (resp. $p(B|A)$). $p_s(A)$ denotes the probability that $A$ defaults due to itself factors, and $p(B|A)$ denotes the probability that $B$ defaults because of the default of $A$.
By considering both factors, we can calculate the node {default probability}. We say a node is more vulnerable if it has higher default probability. The problem is different from the existing research, such as reliability problem and influence maximization problem~\cite{ke2019depth,li2018influence},
which more focus on investigating the reachability for a set of nodes or finding a group of nodes to maximize the influence over the network.
Our problem is of great importance to many real-world applications. Following is a motivating example on financial data analysis.

\begin{figure}[tb!]
  \centering
  \includegraphics[width=0.7\linewidth]{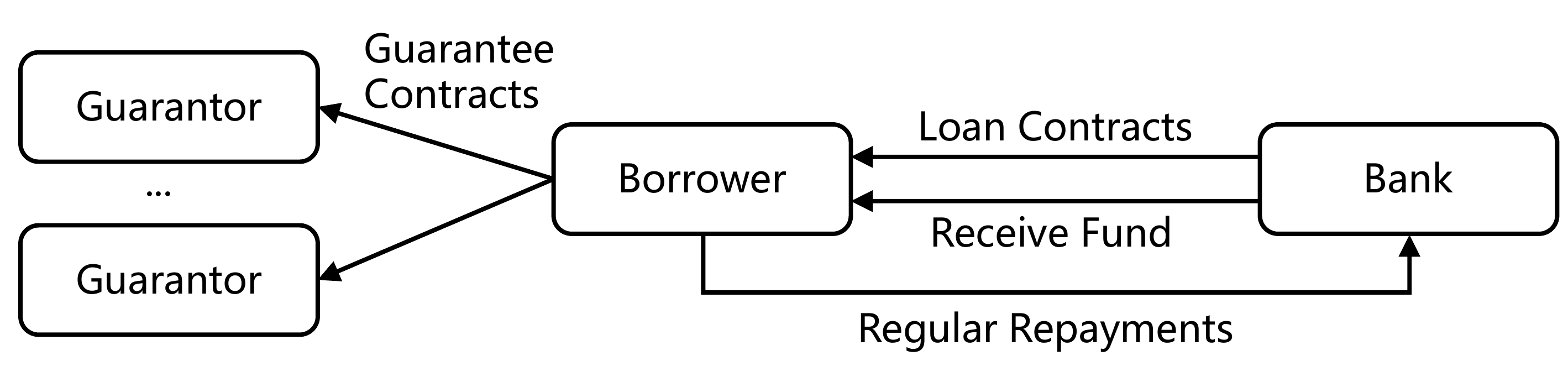}\vspace{-5pt}
  \caption{Guarantee loan process}\label{preliminaries-procedure} \vspace{-3pt}
\end{figure}

\myparagraph{Motivating application} Network-guaranteed loan (also known as guarantee circle) is a widespread economic phenomenon, and attracting increasing attention from the banks, financial regulatory authorities, governments, etc.
In order to obtain loans from banks, groups of small and medium enterprises (SMEs) back each other to enhance their financial security. Figure~\ref{preliminaries-procedure} shows the flow of guarantee loan procedure. When more and more enterprises are involved, they form complex directed-network structures~\cite{cheng2020delinquent}.
Figure~\ref{realnetwork} illustrates a guaranteed-loan network with around $3,000$ enterprises and $7,000$ guarantee relations, where a node represents a small or medium enterprise and a directed edge from node $A$ to node $B$ indicates that enterprise $B$ guarantees another enterprise $A$. Thousands of guaranteed-loan networks of different complexities have coexisted for a long period~\cite{jian2012determinants}. It requires an efficient strategy to prevent, identify and dismantle systematic crises.

\begin{figure}[!tb]
  \centering
  \includegraphics[width=0.7\linewidth]{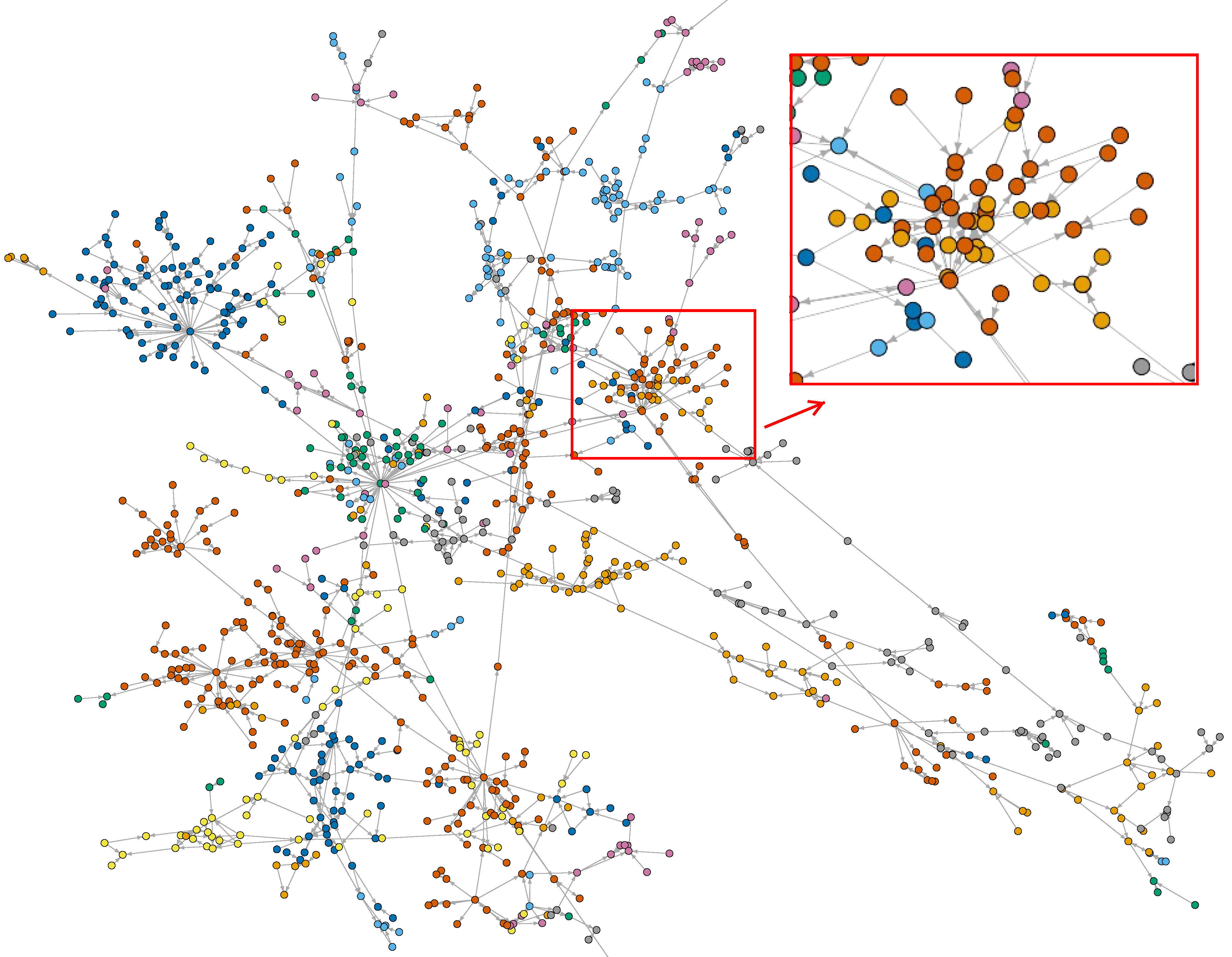}\vspace{-5pt}
  \caption{A real-world guaranteed-loan network}
  \label{realnetwork} \vspace{-3pt}
\end{figure}

Many kinds of risk factors have emerged throughout the guaranteed-loan network, which may accelerate the transmission and amplification of risk. The guarantee network may be alienated from the ``mutual aid group'' as a ``breach of contract''. An appropriate guarantee union may reduce the default risk, but significant contagion damage throughout the networked enterprises could still occur in practice~\cite{mcmahon2014loan}. The guaranteed loan is a debt obligation promise. If one corporation gets trapped in risks, it will spread the contagion to other corporations in the network.
When defaults diffuse across the network, a systemic financial crisis may occur.
It is essential to consider the {contagion damage} in the guaranteed-loan networks.
We can model a guaranteed-loan network with our uncertain graph model, where each node has a self-risk probability and each edge has a risk diffusion probability. Figure~\ref{fig:prob_graph}(a) is toy example of guaranteed-loan network, where Figure~\ref{fig:prob_graph}(e) shows the corresponding relationships between different enterprises and the risk diffusion probabilities.
It is desirable to efficiently identify the $k$ {most vulnerable nodes}, i.e., enterprises with high default risks,
such that the banks or the financial regulatory authorities can pay extra attention to them for the purpose of risk management. It is more urgent than ever before with the slowdown of the economics worldwide nowadays.

In the literature, some machine-learning based approaches (e.g.,~\cite{DBLP:conf/ijcai/ChengTMN019}) have been proposed to predict the node default risk for different applications. For instance, a high-order and graph attention based network representation method has been designed in~\cite{DBLP:conf/ijcai/ChengTMN019} to infer the possibility of loan default events. These approaches indeed consider the structure of networks. However, they cannot properly capture the {uncertain} nature of contagious behaviors in the networks. In the experiment, we also compare with these approaches to demonstrate the advantages of our  model.

\vspace{1mm} \noindent \textbf{Our approach.}
In this paper, to identify the top-$k$ vulnerable nodes, we model the problem with an uncertain graph, and infer the default probability of a node following the possible world semantics,
which has been widely used to capture the contagious phenomenon in real networks~\cite{DBLP:conf/sigmod/AbiteboulKG87,niu2020iconviz,DBLP:conf/ijcai/ChengWZ020,DBLP:conf/kdd/KempeKT03}.
In particular, we utilize
an uncertain graph with two types of probabilities to model the occurrence and prorogation of the default risks in the network.
Note that, we focus on identifying vulnerable nodes for a given network, while the self-risk probabilities and diffusion probabilities can be obtained based on the existing works (e.g.,~\cite{cheng2018prediction,DBLP:conf/ijcai/ChengTMN019}).

\begin{figure}
  \centering
  \includegraphics[width=0.7\linewidth]{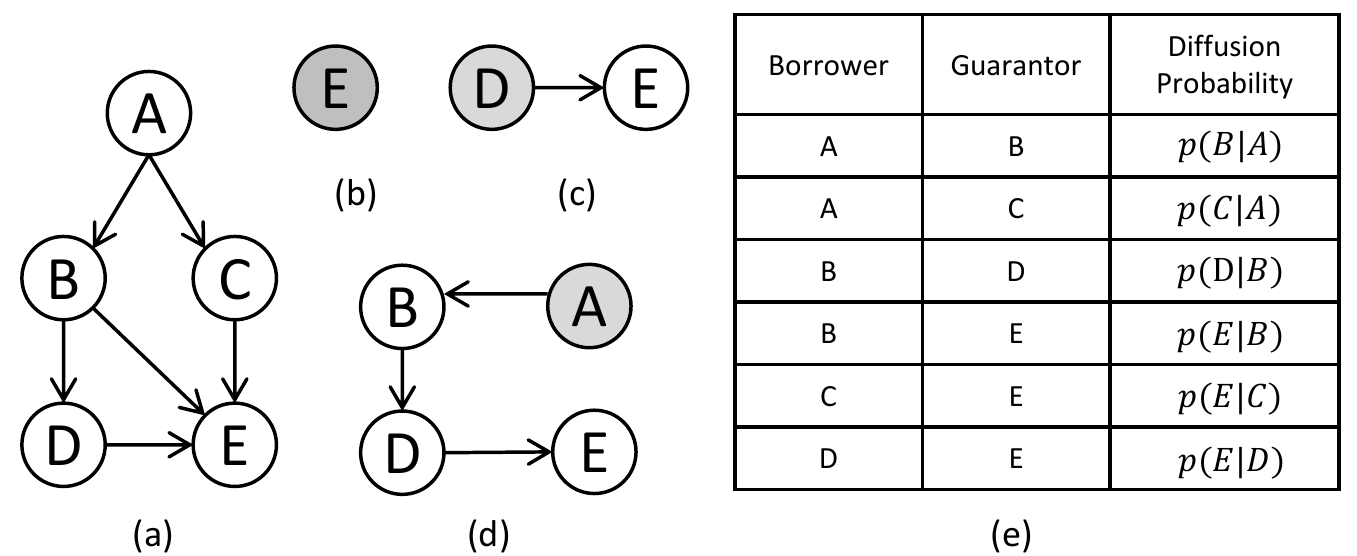}\vspace{-0pt}
  \caption{Example of uncertain guaranteed-loan network}
  \label{fig:prob_graph}\vspace{-0pt}
\end{figure}

Specifically,
Figures~\ref{fig:prob_graph}(a) and~\ref{fig:prob_graph}(e) illustrate the structure of a toy uncertain graph with $5$ nodes and $6$ edges, as well as the associated self-risk probabilities and diffusion probabilities.
Given the probabilistic graph $\mathcal{G}$, we may derive the \textit{default probability} of a node following the possible world semantics, where each possible world (i.e., \textit{instance graph} in this paper) corresponds to a subgraph (i.e., possible occurrence) of $\mathcal{G}$.
Figures~\ref{fig:prob_graph}(b)-(d) are three example possible worlds of  Figure~\ref{fig:prob_graph}(a).
In each possible world, a node exits if it defaults, and an edge $(A,B)$ appears if the default of $A$ indeed leads to the default of $B$.
Taking node $E$ as an example, it may default because of $(i)$ itself, which is represented by a shaded node as shown in Figure~\ref{fig:prob_graph}(b), or because of $(ii)$ the contagion damage initiated by other nodes as shown in Figures~\ref{fig:prob_graph}(c)-(d).
In Section~\ref{sec:preliminary}, we will formally introduce how to derive the default probabilities of nodes.

In this paper, we show that the problem of calculating the default probability of a node alone is \#P-hard,
not mentioning the top-$k$ vulnerable nodes identification problem.
A straightforward solution for the top-$k$ vulnerable nodes computation is to
enumerate all possible worlds and then aggregate the results in each possible world.
However, it is computational prohibitive, since the number of possible worlds of an uncertain graph
may be up to $2^{n+m}$, where $n$ and $m$ are the number of nodes and edges in the graph, respectively.
In this paper, we first show that we can identify the top-$k$ vulnerable nodes by using a limited
number of sampled instance graphs with tight theoretical guarantee.
To reduce the sample size required and speedup the computation, lower/upper bounds based pruning strategies and advanced sampling method are developed.
In addition, to further accelerate the computation, a bottom-$k$ sketch based method is proposed.
To verify the performance in real scenarios, we integrate the proposed techniques with our current loan risk control system, which is deployed in the collaborated bank.

\vspace{1mm}
\noindent \textbf{Contributions.}
The principle contributions of this paper are summarized as follows.

\begin{itemize}
  \item We advocate the problem of top-$k$ vulnerable nodes detection in uncertain graphs, which is essential in real-world applications.

  \item Due to the hardness of the problem, a sampling based method is developed with tight theoretical guarantee.

  \item We develop effective lower and upper bound techniques to prune the searching space and reduce the sample size required. Advanced sampling method is designed to speed up the computation with rigorous theoretical analysis.

  \item A bottom-$k$ sketch based approach is further proposed, which can greatly speedup the computation with competitive results.


  \item We conduct extensive experiments to evaluate the efficiency and effectiveness of our proposed algorithms on $3$ real financial networks and $5$ benchmark networks.

  \item To further verify the advantages of our models in real scenarios, the proposed techniques are integrated into our current loan risk control system, which is deployed in the collaborated bank.
  Through the case study on real-life financial environment, it verifies that our proposed model can significantly improve the accuracy for high-risky enterprises prediction.

\end{itemize}

\vspace{1mm}
\noindent \textbf{Roadmap.} The rest of the paper is organized as follows. Section~\ref{sec:preliminary} describes the problem studied and the related techniques used in the paper. Section~\ref{sec:ours} shows the basic sampling-based method and our optimized algorithms. We report the experiment results in Section~\ref{sec:experiment}.
Section~\ref{sec:dep} introduces the system implementation details and case study.
We present the related work in Section~\ref{subsec:related_work} and conclude the paper in Section~\ref{sec:conclusion}.

\begin{table}\vspace{-0pt}
  \centering
\caption{Summary of notations}\label{tb:notations}
    \begin{tabular}{|c|l|}
      \hline      	
      \textbf{Notation} & \textbf{Definition} \\ \hline \hline
      $\mathcal{G}=(\mathcal{V},\mathcal{E})$ & uncertain graph \\ \hline
      $\mathcal{V}$/$\mathcal{E}$ & node/edge set of $\mathcal{G}$ \\ \hline
      $p_s(v)$ & the self-risk probability of $v$\\ \hline
      $p(v_i|v_j)$ & the diffusion probability\\ \hline
      $p(v)$ & the default probability of $v$\\ \hline
      $h(x)$ & a truly random hash function \\ \hline
      $\mathcal{L}(A,k)$ & the $k$-th smallest hash value of the set $A$  \\ \hline	
      $bk$ & the parameter $k$ in the bottom-$k$ sketch \\ \hline
      $N(v)$ & the set of in-neighbors of node $v$ \\ \hline
      $\mathcal{A}$ &  an approximation algorithm for the problem \\ \hline
      $\mathcal{R}$ &  the set of $k$ nodes returned by $\mathcal{A}$ \\ \hline
      $P^k$ &  the default probability of rank $k$-th nodes \\ \hline
      $(\epsilon,\delta)$ &  the parameters in $(\epsilon,\delta)$-approximation \\\hline
        $t$ &  the sample size \\ \hline
      $p_l(v)$, $p_u(v)$ &  the lower and upper bound of $p(v)$ \\ \hline
      $T_l$, $T_u$ &  the $k$-th largest value of $p_l(v)$ and $p_u(v)$ \\ \hline
      $\mathcal{B}$ &  the set of candidates \\ \hline
      $\mathcal{G}^t$ &  graph by reversing the direction of edges in $\mathcal{G}$ \\ \hline
   \end{tabular}
\vspace{-0pt}
\end{table}

\section{Preliminaries}
\label{sec:preliminary}

In this section, we first we present some key concepts and formally define the problem. Then, we introduce the related techniques used. Table~\ref{tb:notations} summarizes the notations frequently used throughout this paper.

\subsection{Problem Definition}
\label{subsec:key_concept_definition}

We consider an uncertain graph  $\mathcal{G}=(\mathcal{V},\mathcal{E})$ as a directed graph, where $\mathcal{V}$ is the set of nodes and $\mathcal{E}$ denotes the set of edges. $n = |\mathcal{V}|$ (resp. $m = |\mathcal{E}|$) is the number of nodes (resp. edges) in $\mathcal{G}$. Each node $v \in \mathcal{V}$ is associated with a \textbf{self-risk probability} $p_s(v)$, which denotes the default probability of $v$ caused by self-factors. Each edge $(u,v) \in \mathcal{E}$ is associated with a \textbf{diffusion probability} $p(v|u)$, which denotes the probability of $v$'s default caused by $u$'s default.
In this paper, we assume the self-risk probabilities and diffusion probabilities are already available. These probabilities can be derived based the previous studies (e.g.,~\cite{DBLP:conf/ijcai/ChengTMN019}).

For simplicity, when there is no ambiguity, we use uncertain graph, graph and network interchangeably.
In this paper, we derive the default probability of a node by considering both self-risk probability and diffusion probability, which is defined as follows.

\begin{definition}
[Default Probability]
Given an uncertain graph $\mathcal{G}=(\mathcal{V},\mathcal{E})$, for each node $v \in \mathcal{V}$, its default probability, denoted by $p(v)$, is obtained by considering both self-risks probability
and diffusion probability. $p(v)$ can be computed recursively as follows.
\begin{equation}
\label{eqn:dp}
p(v) = 1-(1-p_s(v))\left(\prod_{ ~all~ x \in N(v)} \left(1-p(v|x)p(x)\right)\right)
\end{equation}
where $N(v)$ is the set of in-neighbors of $v$.
\end{definition}

It is easy to verify that the equation above is equal to aggregate the probability over all the possible worlds, i.e.,
\begin{equation*}
p(v) = \sum_{W \in \mathcal{W}} p(W) \times I_W(v)
\end{equation*}
where $\mathcal{W}$ is the set of all possible worlds, $p(W)$ is the probability of a possible world $W$,
and $I_W(v)$ is an indicator function denoting if $v$ defaults in $W$ or not.

\begin{example}
Reconsider the graph in Figure~\ref{fig:prob_graph}. Suppose the self-risk probabilities and diffusion probabilities are all 0.2 for each node and edge. Then, we have $p(A) = 0.2$ and $p(B) = 1- (1-p_s(B))(1 - 0.2 \times p(A))=0.232$.
\end{example}



\myparagraph{Problem statement} Given an uncertain graph $\mathcal{G}=(\mathcal{V},\mathcal{E})$, we aim to develop efficient algorithms to identify the set $\mathcal{R}$ of top-$k$ vulnerable nodes, i.e., the $k$ nodes with the highest default probability.


\begin{theorem}
\label{th:hardness}
It is \#P-hard to compute the default probability.
\end{theorem}

\begin{proof}
We show the hardness of the problem by considering a simple case, where the self-risk probability $p_s(v)$ equals 1 for node $v$, and $p_s(u)$ equals 0 for $u\in \mathcal{V}\setminus v$. Therefore, for the node $u\in \mathcal{V}\setminus v$, the default probability $p(u)$ is only caused by the default of node $v$. Then, the default probability of $p(u)$ equals the reliability from $v$ to $u$, which is \#P-hard to compute~\cite{khan2014fast}. Thus, it is \#P-hard to compute the default probability. The theorem is correct.
\end{proof}

\subsection{Bottom-k Sketch}
In this section, we briefly introduce the bottom-$k$ sketch~\cite{cohen2007summarizing}, which is used in our \ours framework to obtain the statistics information for early stopping condition. Bottom-$k$ sketch is designed for estimating the number of distinct values in a multiset. Given a multiset $A = \{v_1, v_2, \cdots , v_n\}$ and a truly random hash function $h$, each distinct value $v_i$ in the set $A$ is hashed to $(0, 1)$ and $h(v_i) \neq h(v_j)$ for $i \neq j$.
The bottom-$k$ sketch consists of the $k$ smallest hash values,
i.e., $\mathcal{L}(A) = \{h(v_i) | h(v_i) \leq \mathcal{L}(A,k) \wedge v_i \in A\}$,
where $\mathcal{L}(A,k)$ is the $k$-th smallest hash value.
So the number of distinct value can be estimated with
$\frac{k-1}{\mathcal{L}(A,k)}$.
The estimation can converge fast with the increase of $k$, where
the expected relative error is $\sqrt{2/{\pi(k-2)}}$ and the coefficient variation is no more than $1/\sqrt{k-2}$.
To distinguish from the $k$ in the top-$k$ problem, hereafter in this paper, we use $bk$ to denote the parameter $k$ in the bottom-$k$ sketch.

\section{Solutions}
\label{sec:ours}

In this section, we first present a basic sampling method and analyze the sample size required. Then, we introduce the optimized approaches to accelerate the processing.

\subsection{Basic Sampling Approach}
\label{subsec:sampling}

Due to the hardness of computing the default probability, in this section, we propose a sampling based method. In order to bound the worst case performance, rigorous theoretical analysis is conducted about the required sample size.

\myparagraph{Sampling framework}
To compute the default probability, we can enumerate all the possible worlds and aggregate the results. However, the possible world space is usually very large. In previous research, sampling based methods are widely adopted for this case.
That is, we randomly select a set of possible worlds and take the average value as the estimated default probability.
By carefully choosing the sample size, we can return a result with performance guarantee.

Algorithm~\ref{alg:sampling} shows the details of the basic sampling based method. The input is a given graph, where each node/edge is associated with a self-risk/diffusion probability. In each iteration, we generate a random number for each node to determine if it defaults by itself or not (Lines 4-7). Then, we conduct a breath first search from these nodes, i.e., $h_v = 1$, to locate the nodes that will be influenced by them in the current simulation. For each encountered edge, we generate a random number to decide if the propagation will continue or not. For each node, the number of default times is accumulated in Lines 21. The final default probability is calculated by taking an average over the accumulated value $v^c$. Finally, the algorithm returns $k$ results with the largest estimated value.

\begin{example}
Reconsider the example in Figure~\ref{fig:prob_graph}.
Suppose Figures~\ref{fig:prob_graph}(b)-\ref{fig:prob_graph}(d) are the three sampled graphs, and nodes $E,D,A$ are the ones default by itself. Assuming $k=2$, then, nodes $E$ and $D$ are returned results with the largest estimated default probabilities.
\end{example}

\begin{algorithm}[t]
\SetVline 
\SetFuncSty{textsf}
\SetArgSty{textsf}
\small
\caption{Basic Sampling Approach}
\label{alg:sampling}
\Input
{
   $\mathcal{G}:$ a given graph, $k$: a positive integer, $t$: sample size
}
\Output{$\mathcal{R}$: collection of top-$k$ vulnerable nodes}
\Indp
\State{$u_v:= 0$ for all nodes $v \in \mathcal{V}$ }
\For{$i$ in $1$ to $t$}
{
    \State{$h_v:= 0$ for all nodes $v \in \mathcal{V}$}
    \ForEach{$v \in \mathcal{V}$}
    {
      \State{ generate a random number $r_v$ in $[0,1]$}
      \If{ $r_v \leqslant p_s(v)$ }
      {
         \State{$h_v := 1$}
      }
    }

    \State{$\mathcal{Q} \gets \{v|h_v=1\}$ }
    \State{mark $\mathcal{V}\setminus \mathcal{Q}$ as \emph{unvisited} }

    \While{$\mathcal{Q} \neq \emptyset$}
    {
      \State{$v_q \gets \mathcal{Q}.pop()$ }
      \ForEach{$v_a \in N(v_q)$ }
      {
        \If{$v_a$ is \emph{unvisited}}
        {
          \State{generate a random number $r_e$ in $[0,1]$ }
          \If{$r_e > p(v_a|v_q)$}
          {
            \State{continue }
          }
          \State{$h_{v_a}:= 1$ }
          \State{mark $v_a$ as \emph{visited} }
          \State{push $v_a$ to $\mathcal{Q}$ }
        }
      }
    }
    \ForEach{$v \in \mathcal{V}$}
    {
      \State{$v^c:= v^c + h_v$ }
    }
}
\State{$\mathcal{R} \gets$ the top-$k$ nodes ordered by $p_v = v^c/t$ }
\Return{ $\mathcal{R}$ }
\end{algorithm}

\myparagraph{Sample size analysis}
For sampling based methods, a critical problem is to determine the sample size required
in order to bound the quality of returned result.
In this section, we conduct rigorous theoretical analysis about the sample size required. Specifically, we say an algorithm $\mathcal{A}$ is $(\epsilon, \delta)$-approximation if the following conditions hold.

\begin{definition} [{$(\epsilon, \delta)$-approximation}]
Given an approximation algorithm $\mathcal{A}$ for the top-$k$ problem studied, let $\mathcal{R}$ be the set of $k$ nodes returned by $\mathcal{A}$. $P^k$ is the default probability of the $k$-th node in the ground truth order.
Given $\epsilon, \delta \in (0,1)$, we say $\mathcal{A}$ is $(\epsilon, \delta)$-approximation if $\mathcal{R}$ fulfills the following conditions with at least $1-\delta$ probability.
\begin{itemize}
\item [1)] For $v \in \mathcal{R}$, $p(v) \geq P^k - \epsilon$;
\item [2)] For $v \notin \mathcal{R}$, $p(v) < P^k + \epsilon$.
\end{itemize}
\end{definition}

If $\mathcal{A}$ is a $(\epsilon, \delta)$-approximation algorithm for the vulnerable node detection problem, it means that with high probability $(i)$ the default probabilities of returned nodes are at least $P_v^k - \epsilon$; $(ii)$ for the nodes not in $\mathcal{R}$, the default probabilities are at most $P^k + \epsilon$.

\begin{theorem} [Hoeffding Inequality]
Given a sample size $t$ and $\epsilon > 0$, let ${p_v}$ be the unbiased estimation of $p(v)$, where $p_v = \frac{1}{t}\sum_{i=1}^{t}p_v^i$ and $p_v^i \in [a^i, b^i]$. Then we have
\begin{equation}
Pr[{p_v} - p(v) \geq \epsilon]  \leq \exp(-\frac{2t^2\epsilon^2}{\sum_{i=1}^{t}(a^i - b^i)^2})
\end{equation}
\end{theorem}

Based on the Hoeffding inequality, we have following theorem hold.

\begin{theorem} \label{th:middle}
Given the sample size $t$, $\epsilon > 0$ and two nodes $u,v \in \mathcal{V}$, if $p(v) - p(u) \geq \epsilon$, then
\begin{equation*}
Pr[{p_u} - {p_v} > 0] \leq \exp(-t\epsilon^2/2)
\end{equation*}
\end{theorem}

\begin{proof}
We have
\begin{align*}
& Pr[{p_u} - {p_v} > 0] \\
\leq ~& Pr[{p_u} - {p_v} \geq p(u) - p(v) - \epsilon]\\
=~& Pr[{p_u} - {p_v} - (p(u) - p(v)) \geq \epsilon] \\
\leq~& \exp(-\frac{2t^2\epsilon^2}{\sum_{i=1}^{t}2^2})\\
=~&\exp(-t\epsilon^2/2)
\end{align*}
The last two steps consider ${p_u} - {p_v}$ as the estimator of $p(u) - p(v)$. In addition, $p_u^i - p_v^i \in [-1,1]$. Then, we can
feed into the Hoeffding inequality and obtain the final result.
\end{proof}

As shown in Theorem~\ref{th:sample_size}, Algorithm~\ref{alg:sampling} is a $(\epsilon, \delta)$-approximation algorithm if the sample size is no less than Equation~\ref{eq:size}.

\begin{theorem} \label{th:sample_size}
Algorithm~\ref{alg:sampling} is $(\epsilon, \delta)$-approximation if the sample size is no less than
\begin{equation}\label{eq:size}
t = \frac{2}{\epsilon^2}\ln{\frac{k(n-k)}{\delta}}
\end{equation}
where $n$ is the number of nodes, i.e., $|\mathcal{V}|$.
\end{theorem}

\begin{proof}
Suppose we sort the nodes based on their real default probabilities, i.e., $\{v_1,v_2,...,v_n\}$. Then we show the two conditions in $(\epsilon, \delta)$-approximation hold if we have $p_{v_i} - p_{v_j} > 0$ with $p(v_i) - p(v_j) \geq \epsilon$ for $1\leq i \leq k < j \leq n$.
\begin{itemize}
\item For a node $v \in \mathcal{R}$, if $p(v) < P^k -\epsilon < p(v_i) - \epsilon$, it means $p_{v_i} - p_v > 0$ for $i\in [1,k]$. Therefore, $v$ will not be selected into $\mathcal{R}$, which is contradict to the assumption. Thus, the first condition holds.

\item For $v \notin \mathcal{R}$, if $v$ does not belong to the top-$k$ result, the second condition holds naturally. Otherwise, there must be a node $u$ that does not belong to the top-$k$ result being selected into $\mathcal{R}$.
If $p(v) \geq P^k + \epsilon \geq p(u) + \epsilon$, it means $p_v - p_u > 0$. Therefore, $v$ should also be selected into the top-$k$, which is contradict to the assumption. Thus, the second condition holds.
\end{itemize}

Theorem~\ref{th:middle} shows the theoretical result of bounding the order of a pair of nodes.
Since $1\leq i \leq k < j \leq n$, we need to bound the order of $k(n-k)$ pairs of nodes. By applying union bound and Theorem~\ref{th:middle}, we have
\begin{align*}
& \delta = k(n-k)\exp(-t\epsilon^2/2)\\
\Rightarrow & t = \frac{2}{\epsilon^2}\ln{\frac{k(n-k)}{\delta}}
\end{align*}
Therefore, the theorem is correct.
\end{proof}

\subsection{Optimized Sampling Approach}
\label{subsec:optimized}

Based on Theorem~\ref{th:sample_size}, Algorithm~\ref{alg:sampling} can return a result with tight theoretical guarantee. However, it still suffers from some drawbacks, which make it hard to scale for large graphs.
Firstly, to bound the quality of returned results, we need to bound the order of $k(n-k)$ node pairs.
The node size $n$ can be treated as the candidate size, which is usually large in real networks.
Therefore, if we can reduce the size of $n$ (i.e., reduce candidate space) and $k$ (i.e., verify some nodes without estimation), then the sample size can be reduced significantly. Secondly, in each sampled possible world, we only need to determine whether the candidate node can be influenced or not, i.e., compute $h_v$. If the candidate space can be greatly reduced, the previous sampling method may explore a lot of unnecessary space.

According to the intuitions above, in this section, novel methods are developed to derive the lower and upper bounds of the default probability, which are used to reduce the candidate space. In addition, a reverse sampling framework is proposed in order to reduce the searching cost.

\begin{algorithm}[t]
\SetVline 
\SetFuncSty{textsf}
\SetArgSty{textsf}
\small
\caption{Lower Bound Algorithm}
\label{alg:lower}
\Input
{
   $\mathcal{G}:$ a given graph, $z$: the order of bound
}
\Output{ $p_l(v)$:  lower bound of default probability}
\For{$i$ in $1$ to $z$}
{
    \If{$i=1$}
    {
        \State{$p(v) := p_s(v)$ for each $v \in \mathcal{V}$}
        \State{continue}
    }
    \For{each $v$ in $\mathcal{V}$}
    {
      {
      \State{ calculate $p(v)$ by Equitation \ref{eqn:dp} if its in-neighbors' default probabilities have been updated in the previous iteration}
      }
    }
}
\State{ $p_l(v) := p(v)$ for all nodes $v \in \mathcal{V}$}
\Return{ $p_l(v)$ for $v \in \mathcal{V}$ }
\end{algorithm}

\myparagraph{Candidate size reduction}
To compute the lower and upper bounds of the default probability, we utilize the
equation in default probability definition, i.e., Equation~\ref{eqn:dp}.
The idea is that the default probability for each node is in $[p_s(v),1]$ if no further information is given.
By treating each node's default probability as $p_s(v)$ and $1$, we can aggregate the probability over its neighbors to shrink the interval based on Equation~\ref{eqn:dp}. Then, with the newly derived lower and upper bounds for neighbor nodes, we can further aggregate the information and update the bounds.
The details of deriving lower and upper bounds are shown in Algorithms~\ref{alg:lower} and~\ref{alg:upper}.
The algorithms iteratively use the lower and upper bound derived in the previous iteration as the current default probability.
The order of bound denotes the number of iterations conducted.
In each iteration, we only update the bounds of $v$'s default probability if there is any change for the default probability of $v$'s in-neighbors.
Note that, there is a slight difference in the first iteration for the two algorithms, since by using 1 as the default probability will contribute nothing for the pruning.
It is easy to verify that larger order will lead to tighter bounds. Users can make a trade-off between the efficiency and the tightness of bounds.

Given the lower bound and upper bound derived, we can filter some unpromising candidates and verify some candidates with large probability.
Lemma~\ref{le:bound_rule} shows the pruning rules to verify and filter the candidate space.

\begin{lemma} \label{le:bound_rule}
Given the upper and lower bounds derived for each node, let $T_l$ and $T_u$ be the $k$-th largest value in $p_l(v)$ and $p_u(v)$, respectively. Then, we have
\begin{itemize}
\item[1)] For $u \in \mathcal{V}$, $u$ must be in the top-$k$  if $p_l(u) \geq T_u$.
\item[2)] For $u \in \mathcal{V}$, $u$ must not be in the top-$k$  if $p_u(u) < T_l$.
\end{itemize}
\end{lemma}

\begin{proof}
For the first case, suppose a node $u$ with $p_l(u) \geq T_u$ but not being selected in the top-$k$ results, which means a node must have default probability of at least $p_l(u)$ to be selected into the top-$k$ result.
Since $T_u$ is the $k$-th largest value in $p_u(v)$, it means there will be no more than $k$ nodes that satisfy the condition. Therefore, the first case holds.
For the second case, since $T_l$ is the $k$-th largest value of $p_l(v)$, which means $P^k$ must be at least $T_l$. Note that, $P^k$ is default probability of the ranked $k$-th node in the the ground truth order. Therefore, the second case holds.
\end{proof}

\begin{algorithm}[t]
\SetVline 
\SetFuncSty{textsf}
\SetArgSty{textsf}
\small
\caption{Upper Bound Algorithm}
\label{alg:upper}
\Input
{
   $\mathcal{G}:$ a given graph, $z$: the order of bound
}
\Output{ $p_u(v)$:  upper bound of default probability}
\For{$i$ in 1 to $z$}
{
    \For{each vertex $v$ in $\mathcal{V}$}
    {
      \If{$i=1$}
      {
             \State{calculate $p(v)$ by treating its in-neighbors' default probabilities as 1}
      }
      \Else{
      \State{calculate $p(v)$ if its in-neighbors' default probabilities have been updated in the previous iteration}
      }
    }
}
\State{ $p_u(v) := p(v)$ for all nodes $v \in \mathcal{V}$}
\Return{ $p_u(v)$ for $v \in \mathcal{V}$ }
\end{algorithm}

\begin{algorithm}[t]
\SetVline 
\SetFuncSty{textsf}
\SetArgSty{textsf}
\small
\caption{Candidate Reduction}
\label{alg:bound_candidate}
\Input
{
   $p_u(v) / p_l(v)$: upper and lower bound for each node,
   $k$: a positive integer
}
\Output{ $\mathcal{B}$: candidates selected, $k'$: the number of nodes verified }
\State{$T_l$ $\leftarrow$ the $k$-th largest value in $p_l(v)$ }
\State{$T_u \leftarrow$ the $k$-th largest value in $p_u(v)$ }
\State{$\mathcal{B} := \emptyset$; $k' = 0$ }
\For{each $v$ in $\mathcal{V}$}
{
    \If{$p_l(v) \geq T_u$}
    {
        \State{$k'++$}
        \State{insert $v$ into the result set}
        \State{continue}
    }

    \If{ $p_u(v) \geq T_l$}
    {
        \State{push $v$ into $\mathcal{B}$ }
    }
}
\Return{$\mathcal{B}$ and $k'$ }
\end{algorithm}

\begin{algorithm}[t]
\SetVline 
\SetFuncSty{textsf}
\SetArgSty{textsf}
\small
\caption{Reverse Sampling Algorithm}
\label{alg:reverse}
\Input
{
   $\mathcal{G}^t:$ a given graph by reverse the direction each edge, $\mathcal{B}$: candidate nodes
}
\Output{$h_v$: for each node $v \in \mathcal{B}$ in one sample}
    \State{$h_v:= 0$ for all nodes $v \in \mathcal{V}$ }
    \ForEach{node $v$ in $\mathcal{B}$}
    {
      \State{mark all nodes as \emph{unvisited} }
      \State{$\mathcal{Q} \gets \{v\}$ }
      \While{$\mathcal{Q} \neq \emptyset$}
      {
        \State{$u = \mathcal{Q}.pop()$ }
        \If{$h_u =1$}
        {
          \State{$h_v:= 1$ and break }
        }
        \If{$u$ is \emph{unchecked}}
        {
          \State{generate a random number $r_u$ in $[0,1]$ }
          \State{mark $u$ as \emph{checked} }
          \If{ $r_u \leqslant p_s(u)$ }
          {
            \State{$h_u: =1, h_v: = 1$ and break }
          }
          \ForEach{$u' \in N(u)$ }
          {
            \If{$(u',u)$ is \emph{unchecked}}
            {
              \State{generate a random number to mark $(u',u)$ as survived or not }
            }
          }
        }
        \State{mark $u$ as \emph{visited} }
        \ForEach{$u' \in N(u)$ }
        {
          \If{$u'$ is \emph{unvisited} and $(u',u)$ is \emph{survived}}
          {
            \State{push $u'$ into $\mathcal{Q}$ }
          }
       }
      }
    }
\Return{ $h_v$ for nodes in $\mathcal{B}$ }
\end{algorithm}

Based on Lemma~\ref{le:bound_rule}, Algorithm~\ref{alg:bound_candidate} shows the details of reducing candidate space. The algorithm takes the derived lower and upper bounds as input and outputs the candidate nodes $\mathcal{B}$ and the number $k'$ of verified nodes. The verified $k'$ nodes will be put into the result set directly.
Note that, if we can verify $k'$ nodes based on the first pruning rule, then we only need to find top-$(k-k')$ nodes from the candidate set $\mathcal{B}$. In this case, we reduce both the value $k$ and $n$ of Equation~\ref{eq:size} to $k-k'$ and $|\mathcal{B}|$, respectively.

\myparagraph{Reverse sampling approach}
Based on Algorithm~\ref{alg:bound_candidate}, we can greatly reduce the candidate space, which performance is verified in our experiments on real-world datasets. In the basic sampling method, it aims to estimate the default probability for each node. Here, we only need to compute the probability for the candidate nodes. Especially, when the candidate size is small, the previous sampling method will explore a lot of unnecessary space.

Intuitively, given a sampled possible world,
for each candidate node, we only need to verify if it can be reached by a node with $h_v = 1$. Therefore, we can conduct a reverse traversal from the candidate nodes to see if it can meet the criteria. The details are shown in Algorithm~\ref{alg:reverse}, where $\mathcal{G}^t$ is the graph by reversing the direction of each edge in $\mathcal{G}$. Note that, our reverse sampling method is different from the reverse sampling framework used in influence maximization problem~\cite{Borgs:2014}.

The inputs are the graph $\mathcal{G}^t$ and candidate set $\mathcal{B}$. After processing a sample, it returns $h_v$ for each node $v$ in $\mathcal{B}$. At first, we set $h_v = 0$ for all the nodes. Then we conduct a breath first search from each node in the candidate set.
For each encountered node and edge, we mark it as \emph{checked} and store the corresponding information (e.g., \emph{survived} and $h_v$) in order to avoid generating random numbers for the same node/edge multiple times.
The BFS terminates if it encounters a node $h_v$ with $h_v=1$ or there is no more node to be explored (Lines 6-8). If it encounters a node with $h_v=1$, it means the candidate node is influenced, and vice versa. Through this way, we can greatly reduce the computation cost by filtering unnecessary searching space. Given sample size $t$, we repeat the process $t$ times and accumulate the $h_v$ value to estimate the default probability.

\myparagraph{Reverse sampling based method}
By integrating the pruning strategies and the reverse sampling technique, we have the {reverse sampling based algorithm}.
According to Theorem~\ref{th:reverse_size}, the approach is $(\epsilon,\delta)$-approximation if the sample size fulfills Equation~\ref{eq:size_bound}.


\begin{theorem} \label{th:reverse_size}
The reverse sampling based algorithm is $(\epsilon, \delta)$-approximation if the sample size is larger or equal than
\begin{equation}\label{eq:size_bound}
t = \frac{2}{\epsilon^2}\ln{\frac{(k-k')(|\mathcal{B}|-k+k')}{\delta}}
\end{equation}
\end{theorem}

Note that, we use the reverse sampling method here because the candidate space is reduced. In addition, it only accelerate the computation of influenced nodes in a sampled possible world. Based on the developed pruning techniques, we can verify some results and filter some candidates. To bound the quality of returned results, we need to bound the order of $(k-k')(|\mathcal{B}|-k+k')$ pairs.

\subsection{Bottom-k Based Approach}
\label{subsec:bottomk}

Based on the lower and upper bounds derived, we can reduce the candidate space. In addition, by using the reverse sampling technique, we can reduce the cost of exploring samples.
The reverse sampling based algorithm can return a result with tight theoretical guarantee, which reduces the sample size from Equation~\ref{eq:size} to Equation~\ref{eq:size_bound}.
However, in many real cases, the sample size and computation cost is still large.
Intuitively, we only need sufficient samples to obtain a competitive result. In this section, we derive a method by using the bottom-$k$ technique, which can greatly accelerate the procedure with competitive top-$k$ results. We first present the idea of finding the top-1 result. Then, we extend the idea for the top-$k$ scenario.

\myparagraph{Finding the top-1 result}
In the reverse sampling approach, when we process the samples one by one. We can terminate the processing, if there is a node that has sufficient statistic.
In this paper, we use bottom-$k$ sketch to serve this role. The idea is that, we first apply the lower and upper bound technique to obtain $k'$ and $\mathcal{B}$. Let $t$ be sample size computed by using Equation~\ref{eq:size_bound}. We assign each sample an id and generate a random hash value in $(0,1)$ for each of them.
Since we does not materialize the samples, the time complex of generating hash value is only $O(t)$.
We sort the samples in ascending order based on the hash value, and materialize the samples accordingly by using the reverse sampling framework. For each node $v$ in the candidate set, we record a accumulated value $v^c$. Based on Theorem~\ref{le:bk}, the node whose $v^c$ reaches $bk$ first is the top-1 result. $bk$ is the threshold selected.

\begin{theorem} \label{le:bk}
The node selected by using the above procedure is the top-1 node.
\end{theorem}

\begin{proof}
Suppose node $u$ is the first node that reaches the criteria and the hash value of its $bk$-th encountered sample is $h^{bk}(u)$. According to the property of bottom-$k$ sketch, we can estimate the default probability $p(u)$ with $\frac{bk-1}{h^{bk}(u)t}$. If $v$ is the second node that reaches the criteria. We must have $h^{bk}(v) > h^{bk}(u)$. Therefore, the corresponding estimated value is smaller that of $u$. The theorem is correct.
\end{proof}

Here, we use $bk$ to measure if the statistic is sufficient or not.
Even though the bottom-$k$ based method does not offer tight theoretical guarantee
as the previous approaches. Through our experimental evaluation, the bottom-$k$ based method shows great advantage compared with the others.

\begin{example}
Reconsider the graph in Figure~\ref{fig:prob_graph} with $k=1$ and $bk=2$.
Suppose nodes $D$ and $E$ are the candidates in $\mathcal{B}$.
Then, we only need to reverse the graph and conduct the traversal from $D$ and $E$.
For the simplicity, we do not present the reverse graph here.
Figure~\ref{fig:prob_graph}(b) is the first reverse sample. That is, $E$ is default by itself, and $D$ does not meet any default nodes when back-traversing from $D$. Therefore, the counter of $E$ is set as 1. Figure~\ref{fig:prob_graph}(c) is the second sample. Then, the counter of $E$ becomes 2, which meets the bottom-$k$ criterion. Thus, $D$ is the top-$k$ result returned.
\end{example}

\myparagraph{Finding the top-$k$ results}
To find the top-$k$, we follow similar manner as that in finding the top-1 result.
By extending Theorem~\ref{le:bk}, we can stop exploring the samples when there are $k-k'$ nodes with sufficient statistic. That is, we continue visiting the samples until there are $k-k'$ nodes with counters reaching $bk$.
Note that, there may be case when the stop condition cannot be met after all the samples are processed. Then, the algorithm turns to the reverse based sampling method, and we just return the $k-k'$ nodes with the largest estimated value. While, according to the experiments over real-world datasets, the algorithm can coverage quickly with $bk$.

\myparagraph{Complexity analysis}
In this paper, there are two types of samples involved, i.e., basic sample (Algorithm~\ref{alg:sampling}) and reverse sample (Algorithm~\ref{alg:reverse}). We use $t_b$ and $t_r$ to denote the cost of generating a basic sample and a reverse sample, respectively. To obtain a sample, in the worst case, we need to traverse the whole graph once for both sampling approaches, which cost is $O(m+n)$. Therefore, for the basic sampling approach (i.e., Algorithm~\ref{alg:sampling}), the time complexity is $O((m+n)\frac{2}{\epsilon^2}\ln{\frac{k(n-k)}{\delta}} + n\log{n})$, where $\epsilon$  and $\delta$ are the input parameters. For the reverse sampling based method, the bound computation and candidate reduction phase cost $O(z(m+n))$ and $O(n\log{n})$, respectively. Thus, the complexity is $O(t_r \frac{2}{\epsilon^2}\ln{\frac{(k-k')(|\mathcal{B}|-k+k')}{\delta}} + z(m+n) + 2n\log{n})$.
The time complexity of the bottom-$k$ based method is the same as that of the reverse sampling based approach. This is because, we need to explore all the samples in the worst case, but it is usually much faster than others in practice.

\section{Experiment}
\label{sec:experiment}

In this section, we conduct extensive experiments to evaluate the effectiveness and efficiency of our proposed methods.

\begin{table}
\centering
\caption{Details of experimental datasets} \label{tb:dataset}\vspace{-0pt}
\renewcommand{\arraystretch}{1.1}
\setlength{\tabcolsep}{7pt}
  \begin{tabular}{lrrrrr}
  \toprule
  Datasets & \# Nodes & \# Edges & Avg Deg. & Max Deg. \\
  \midrule
  Bitcoin &  3,783 & 24,186 & 6.39 & 888\\
  Facebook &  4,039 & 88,234 & 21.85 & 1,045\\
  Wiki &  7,115 & 103,689 & 14.57 & 1,167\\
  P2P &  62,586 & 147,892 & 2.36 & 95\\
  Citation &  2,617 & 2,985 & 1.14 & 44\\
  Interbank &  125 & 249 & 1.99 & 47\\
  Guarantee &  31,309 & 35,987 & 1.15 & 14,362\\
  Fraud &  14,242 & 236,706 & 16.62 & 85,074\\
  \bottomrule
\end{tabular}\vspace{-0pt}
\end{table}

\subsection{Experiment Setup}

\myparagraph{Datasets} We conduct the experiments on 3 real-world financial datasets, i.e., Interbank\footnote{\url{https://github.com/carloscinelli/NetworkRiskMeasures}}, Fraud and Guarantee and 5 public benchmark datasets with drastically varying sizes and characteristics. Fraud and Guarantee are our contributed dataset. The details about the 3 real-world financial datasets are shown as follows.

\begin{itemize}

  \item \textbf{Interbank}. Interbank networks is generated by the maximum-entropy (ME) approach \cite{anand2015filling}, in which each node represents a bank and edge corresponds to an interbank loan from the lender bank to the borrow bank.

  \item \textbf{Fraud}. Credit card fraud networks with $19,240$ nodes and $34,892$ edges is constructed based on credit card fraud transactions of a major commercial bank. Each edge represents a trade between the consumer and merchant.

  \item \textbf{Guarantee}. The guaranteed loans network is from a major commercial bank spanning 4 years. The names of the customers in the records are encrypted and replaced by IDs. We can access the guarantee relationships, which denotes an edge between the guarantor to borrower. Besides, in case studies, we also get the basic profile information such as the enterprise scale, and loan information such as the loan credit.
\end{itemize}

Besides the real-world financial datasets, we also employ 5
benchmark datasets, which are public available. We download Citation from network repository\footnote{\url{http://networkrepository.com/}}. The others are downloaded from SNAP\footnote{\url{https://snap.stanford.edu/data/}}. The statistic details of datasets are shown in Table \ref{tb:dataset}.

\begin{figure}[!tbp]
  \centering
  \begin{minipage}[b]{0.4\textwidth}\vspace{5pt}
    \includegraphics[width=\textwidth]{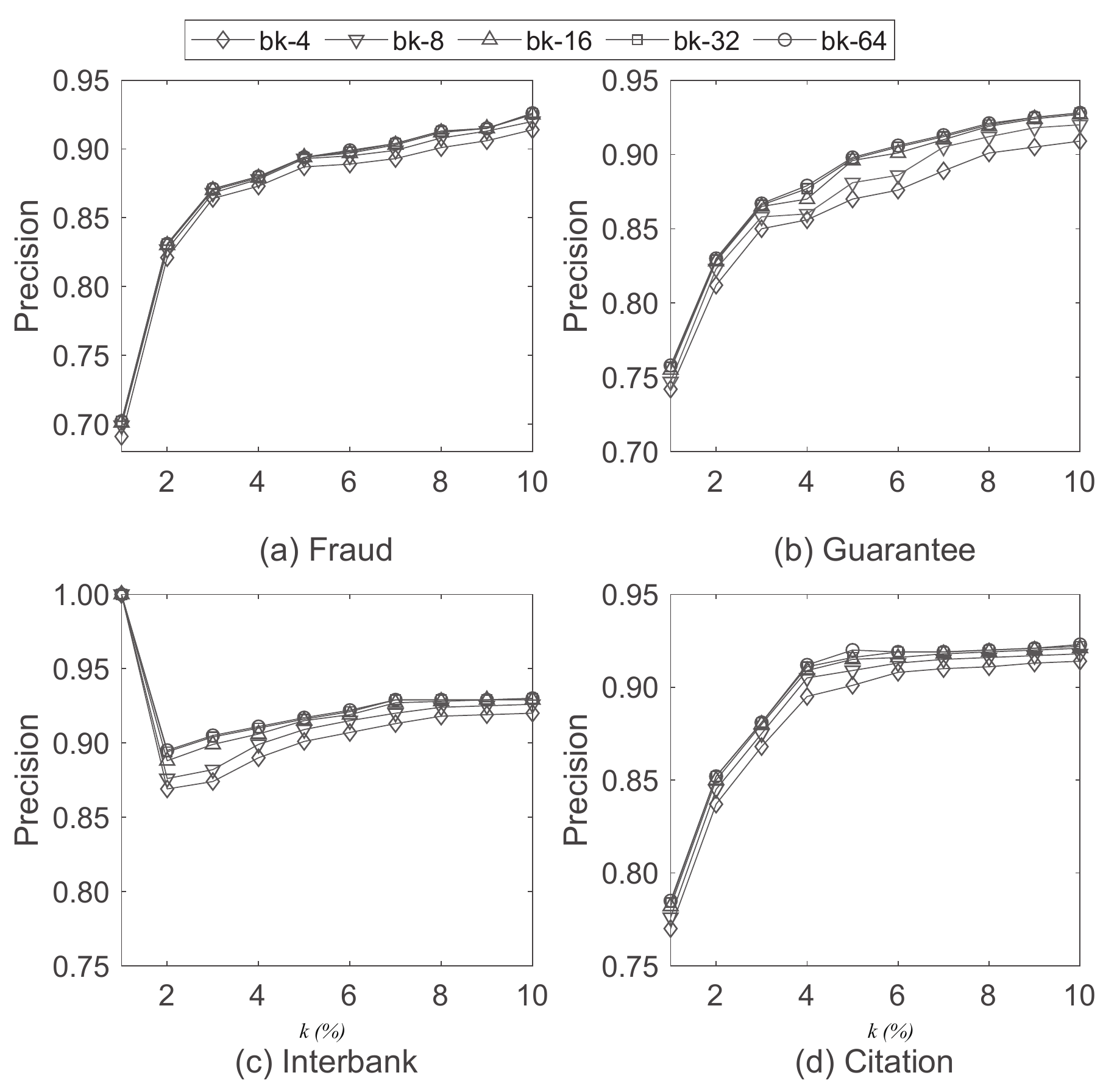}
    \caption{Parameter $bk$ tuning for bottom-$k$ based method}\label{tuning_bs}
  \end{minipage}
  \ \ \ \
  \begin{minipage}[b]{0.4\textwidth}\vspace{-5pt}
    \includegraphics[width=\textwidth]{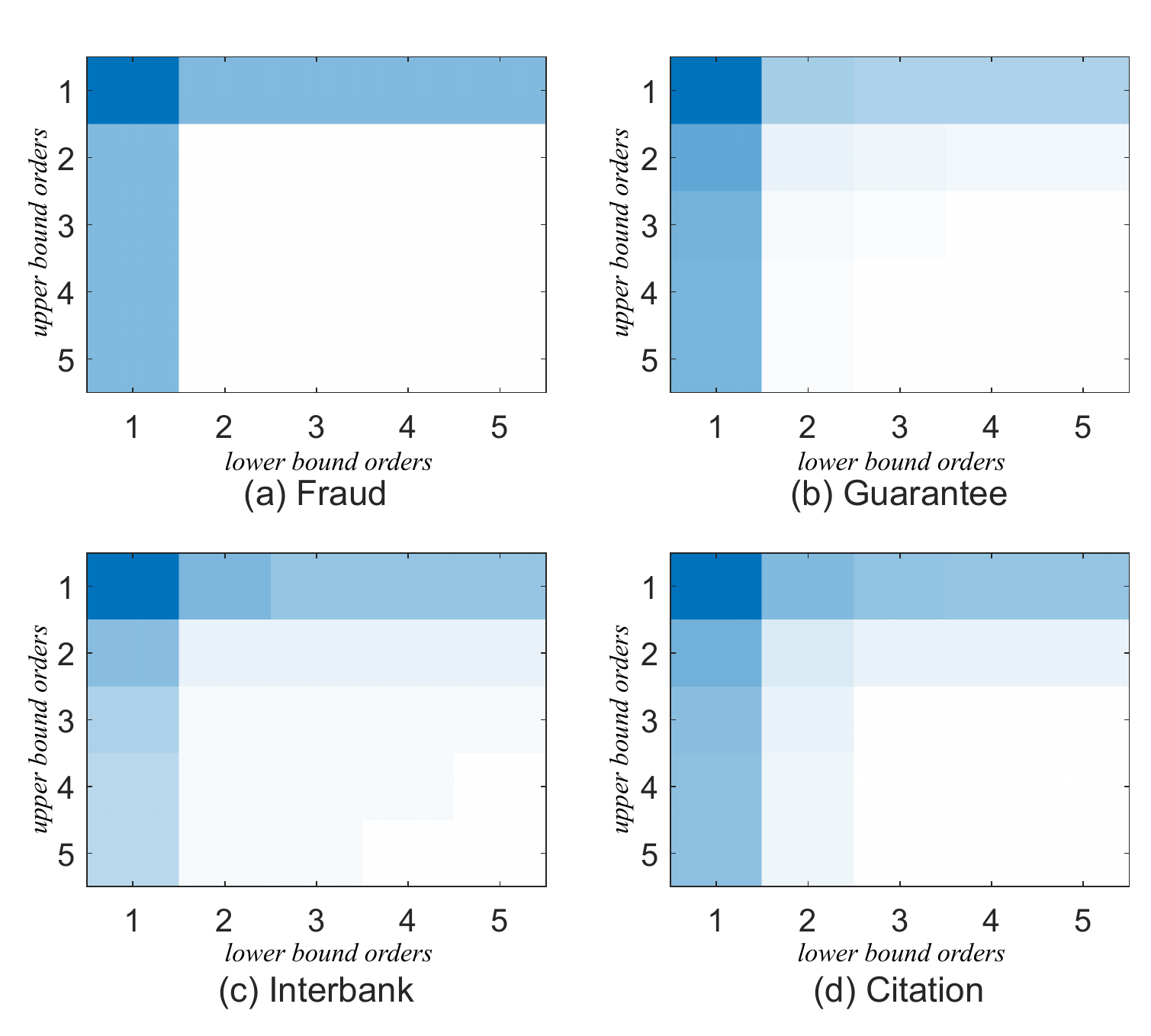}
    \caption{Parameter tuning for the order of bounds}\label{tuning_bound}
  \end{minipage}
\end{figure}

%
%

\begin{figure*}[tb!]\vspace{5pt}
  \centering
  \includegraphics[width=0.97\linewidth]{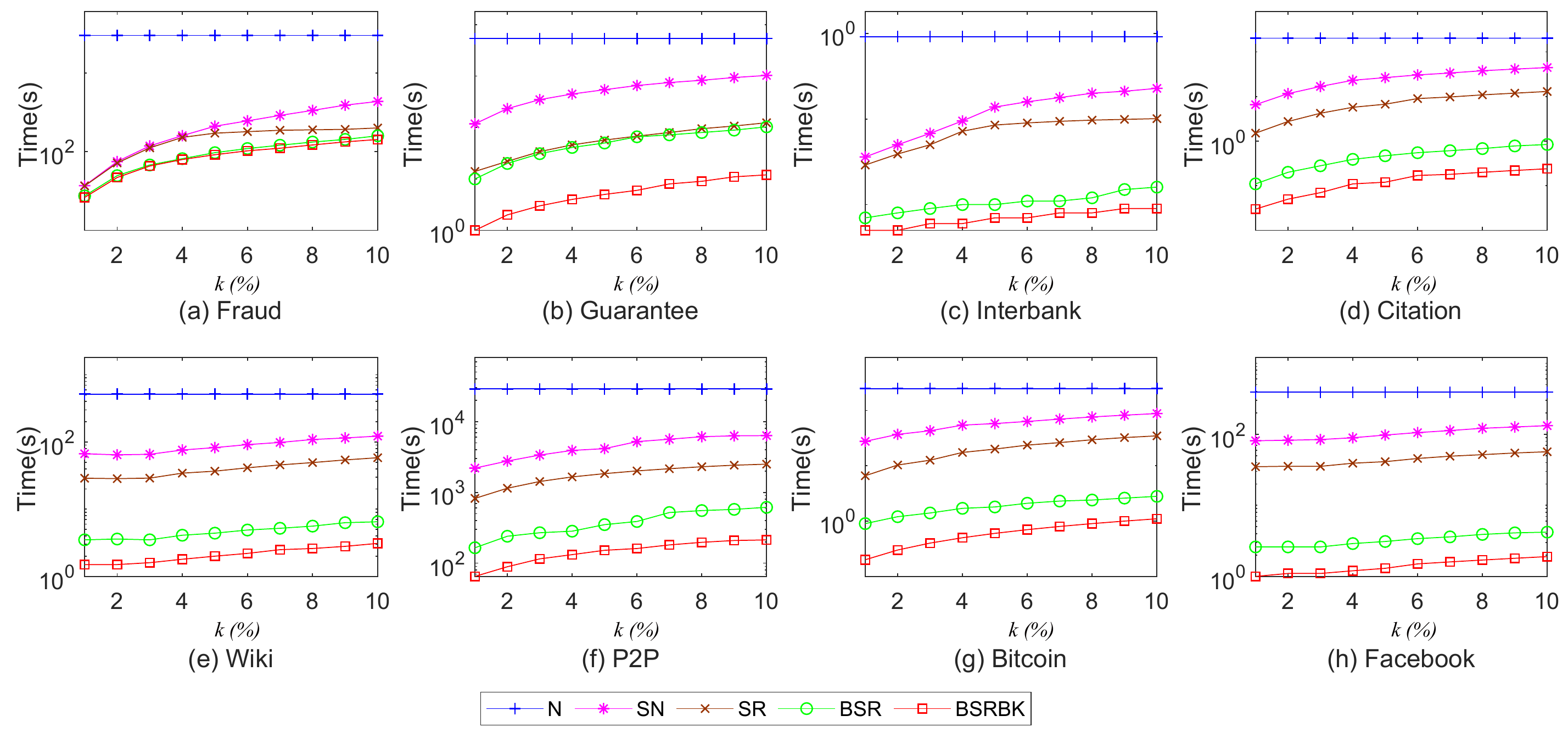}\vspace{-5pt}
  \caption{Efficiency evaluation}\label{exp_times}\vspace{-5pt}
\end{figure*}

\myparagraph{Algorithms}
We evaluate the following algorithms to demonstrate the performance of following algorithms.
\begin{itemize}
  \item N (\textit{Naive}). Algorithm~\ref{alg:sampling} with fixed sample size.
  \item SN (\textit{Naive+Sample}). Algorithm~\ref{alg:sampling} with the sample size calculated by Equation \ref{eq:size}.
  \item SR (\textit{Sample+Reverse}). Algorithm that uses reverse sampling method with candidate set derived with second rule of Lemma 1.
  \item BSR (\textit{Bound+Sample+Reverse}). Optimized sampling method by integrating reverse sampling and bounds filtering techniques with the sample size calculated by Equation~\ref{eq:size_bound}.
  \item BSRBK (\textit{Bound+Sample+Reverse+Bottom-k}).
  Bottom-$k$ based method by integrating reverse sampling and bounds filtering techniques.
\end{itemize}

\myparagraph{Parameters and workload}
To evaluate the effectiveness of proposed techniques, the precision is reported.
For the ground truth, we use 20000 sampled possible worlds to obtain the results, which setting is commonly used in related research~\cite{DBLP:conf/kdd/KempeKT03,li2018influence,ke2019depth}.
For Fraud and Guarantee datasets, the self-risk and diffusion probability are obtained in our previous research \cite{fu2016credit,cheng2018prediction}. For the other datasets, the probability is randomly selected from $[0,1]$. For parameter $k$, we vary it from $1\%|\mathcal{V}|$ to $10\%|\mathcal{V}|$, where $|\mathcal{V}|$ is the corresponding graph node size. We set $\epsilon = 0.3$ and $\delta = 0.1$
for computing the sample size.

In the experiment, we run all the methods by a server with two Intel E5-2680 CPU, 512GB memory and CentOS v7.2 operation system. In the system implementation, we develop the web backend server by Spring Cloud and frontend with JavaScript D3.js library. In case studies, CNN-max, crDNN and HGAR are experimented on a GPU server with two pieces of Telsa-P100 and implemented by Tensorflow.

\subsection{Parameter Tuning}
In this section, we tune the parameters $bk$ and the order of bounds on 4 datasets, i.e., Citation, Interbank, Fraud and Guarantee.

\myparagraph{Tuning the parameter \textit{bk}}
As analyzed in the paper, the precision of \ours should converge rapidly with the increase of $bk$. We vary $bk$ from 4 to 64.
The results are shown in Figure~\ref{tuning_bs}.
Note that $bk$-$X$ means $bk$ is set to $X$.
With the increase of $bk$, the algorithm converges quickly for all the datasets.
When $bk$ reaches 8, the drop of performance already becomes less significant.  Thus, in the following experiments, we set $bk$ to 16.

\myparagraph{Tuning the order of bounds}
Since the tightness of lower and upper bounds may greatly affect the sample size and computation cost, we conduct the experiments to tune the order of bounds.
We vary the order of bounds from 1 to 5 and set $k$ as $5\%$ of the number of nodes. The candidate size is reported.
Figure \ref{tuning_bound} visualizes the result with heatmaps. The lighter the color is, the less number of candidates will be. As we can see, the candidate size decreases rapidly at the beginning, and reach steady when the order reaches 2 for most cases. Therefore, we set the order of upper and lower bounds to 2 for the following experiments.

\begin{figure}[tb!]
  \centering
  \includegraphics[width=0.55\linewidth]{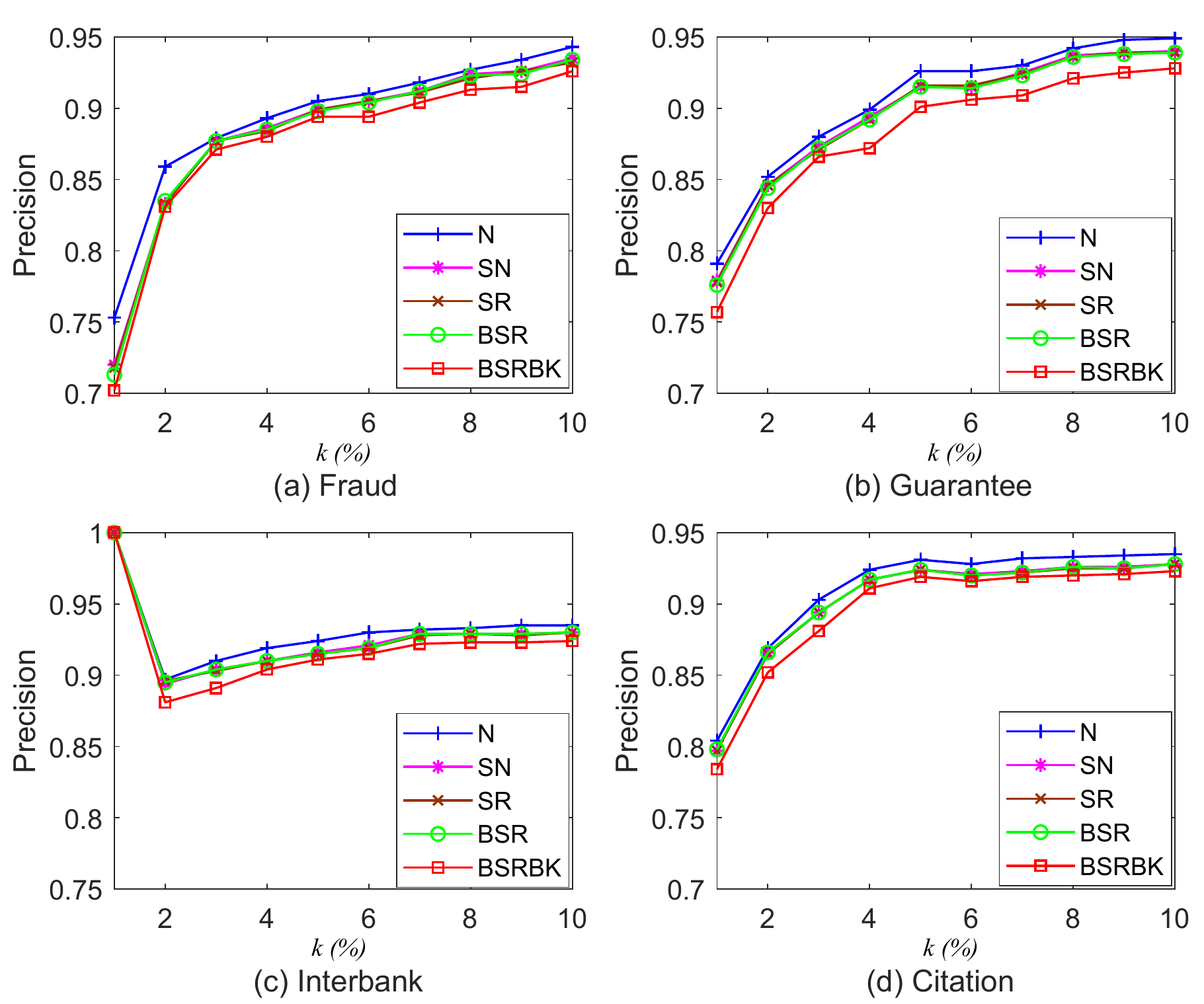}\vspace{-5pt}
  \caption{Effectiveness evaluation}\label{exp_precision}
\end{figure}

\subsection{Efficiency Evaluation}
To demonstrate the efficiency of proposed techniques, we conduct experiments on all the datasets and report the response time. The results are shown in Figure \ref{exp_times}. In all methods, the computation time gradually increases along with $k$ except for the naive approach N, because N uses a large fixed sample size.
For the other methods, the sample size may change when $k$ increases.
As we can observe, algorithm N is the most time-consuming method, and the algorithm runs faster when more accelerating techniques involved. SR is better than SN because the reverse sampling technique and candidate set derived can greatly reduce the sampling cost.
BSR is better than SR, since we can reduce the candidate space and sample size by using the lower and upper bounds derived.
\ours is better than BSR
because of the novel stop condition used. \ours always outperforms
the others and achieves up to 100x acceleration. These observations strongly proves the advantage of proposed techniques.

\subsection{Effectiveness Evaluation}

To evaluate the effectiveness of proposed methods, the precision is reported by varying $k$ from $1\%|\mathcal{V}|$ to $10\%|\mathcal{V}|$. The results are shown in Figure~\ref{exp_precision}.
Generally, the precision of the 5 methods is very close to each other,
and the largest gap between the naive method N and \ours is only 3\%.
Compared with the speedup in efficiency, the precision difference is much less noticeable.
The naive method N is slightly better than the other methods, because it has used more samples. SN, SR and BSR report almost the same result, because they obtain the same theoretical guarantee.
It should be noted that for the Interbank dataset, $1\%|\mathcal{V}| = 1$ and all methods successfully detect that node. Therefore, the precision is 1 as shown in Figure~\ref{exp_precision}(c).
As observed, the experiment results prove that \ours could achieve significant performance acceleration while keeping a tolerable precision reduction.

\begin{figure*}[tb!]
  \centering
  \includegraphics[width=0.9\linewidth]{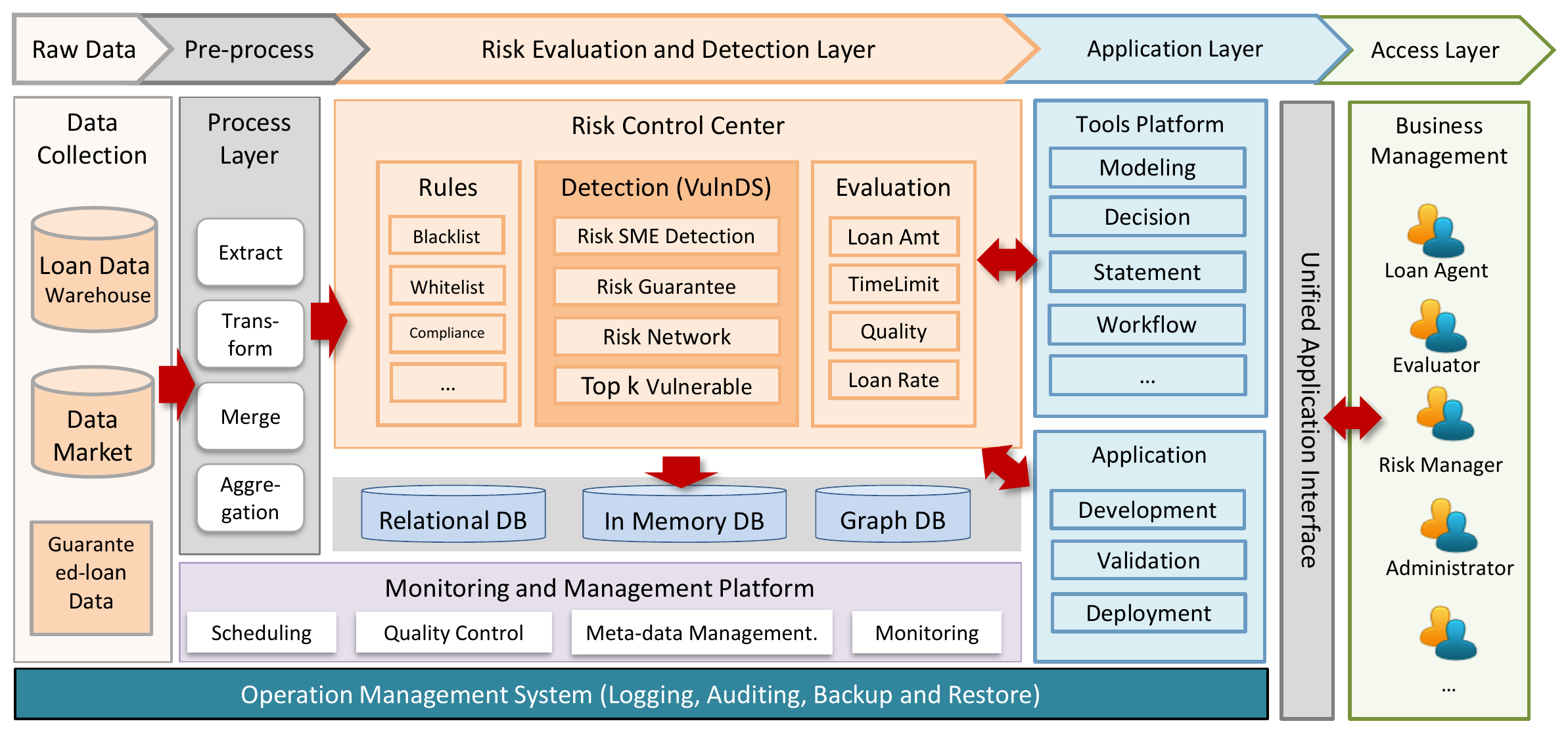}
  \caption{System architecture of {VulnDS} in a loan management system}\label{dep_arc}
\end{figure*}

\section{Implementation and Case Study}
\label{sec:dep}

In this section, We present a system, named {VulnDS}, by integrating the proposed techniques with our current loan risk control system. The control system is deployed in our collaborated bank, which can evaluate our methods in real scenarios.
We first present the overall architecture for {VulnDS}, and then describe the details of system implementation. Finally, we report the interface, observation and case study after system deployment.

\subsection{System Implementation and Deployment}

\myparagraph{Architecture Overview}
Figure \ref{dep_arc} shows the architecture overview of the {VulnDS} in a loan management system. We collect origin data from three upstream: loan data warehouse, data market and external loan data. In the pre-processing layer, raw records are extracted, merged and aggregated for risk control. We employ the in-memory database to store the frequent queried data, and graph database to preserve networked relationships, as well as rational DB for conventional tables. We utilize a monitoring platform for scheduling submitted tasks from the pre-processing and risk control module. The risk assessment results are consumed by the tools and application platform, which is the main scenario to control loan risks. Different roles of business users access the system from a unified application interface.

The risk control center consists of three main parts: the rule engine, vulnerable detection system and evaluation module. Rule engine mainly includes loan blacklist, white list and compliance rules.
If a loan passes the rule check, it will be then processed by our proposed vulnerable detection system. {VulnDS} assess the self-risk of SME, the risk of guarantee relationships, and detect the top-$k$ vulnerable nodes by our methods. Evaluation module leverage the output of {VulnDS} to quantify the loan grant amount, time limit and interest ratio, etc. Once the bank issue a loan, post-loan process are activated immediately. All three steps in the risk control center will be employed to evaluate all issued loans regularly. In our implementation, we detect all loans monthly by the proposed {VulnDS} in a risk control center.

\begin{figure}[tb!]\vspace{-5pt}
  \centering
  \includegraphics[width=0.65\linewidth]{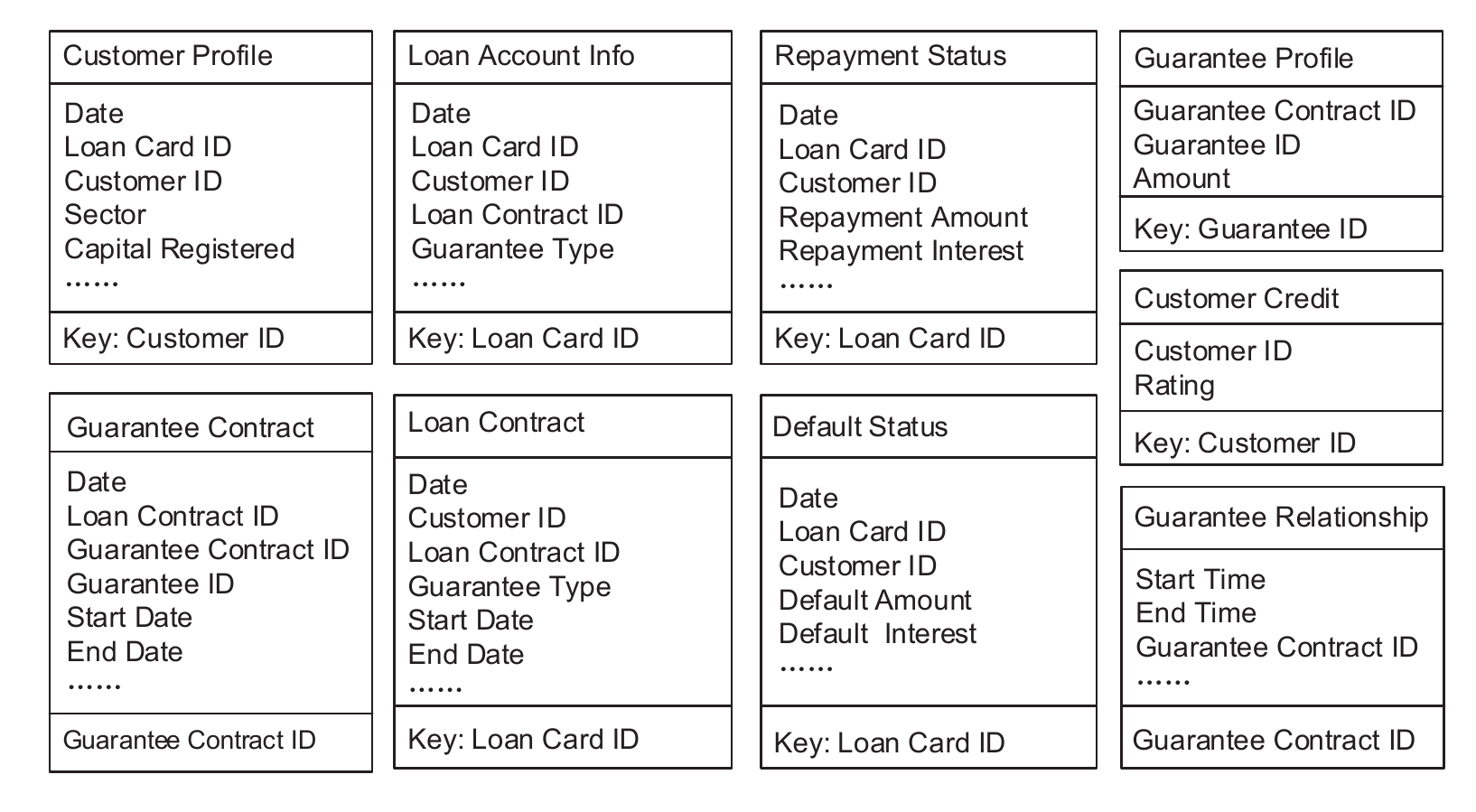}\vspace{-5pt}
  \caption{Overview of data association}\label{preliminaries-rawdata}
\end{figure}

\begin{figure}[tb!]
  \centering
  \includegraphics[width=0.8\linewidth]{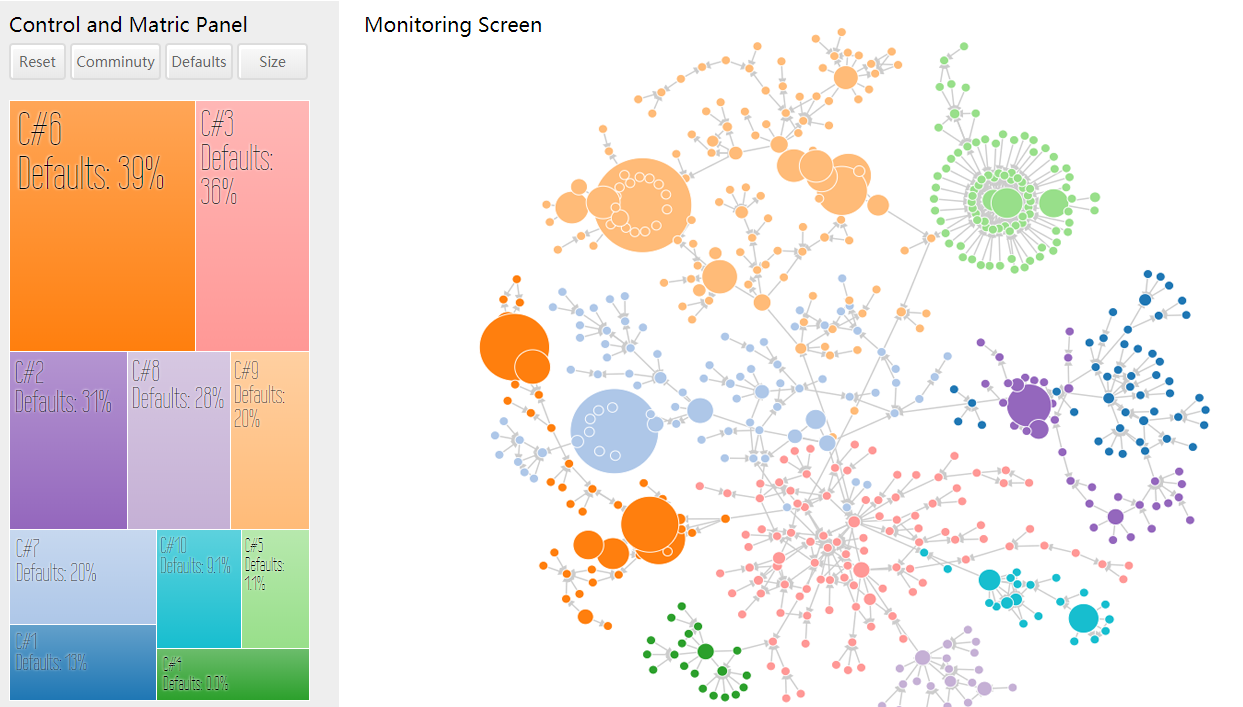}
  \caption{UI of deployed loan management system}\label{deployment}
\end{figure}

\myparagraph{Implementation Details}
Figure \ref{preliminaries-rawdata} shows an overview of the data association, which is extracted by the pre-processing layer.
We employ the internal black and white lists from our collaborated bank. The rules are mainly under the compliance of the new Basel protocol\cite{engelmann2011basel}. In vulnerable detection system, we employ HGAR \cite{DBLP:conf/ijcai/ChengTMN019} for self-risk assessments, p-wkNN \cite{cheng2018prediction} to infer the probability of risk guarantee relationships. The proposed methods are utilized for the final vulnerable SME detection. During implementation, we use the Drools \cite{thu2017transforming} on Apache Flink as the rule engine, in which the hot data are stored in Redis \cite{patel2015sales}. We employ neo4j as the graph database, visualize the graph by open-source software package D3.js and layout ForceAtlas2 \cite{jacomy2014forceatlas2}.  The training model and system implementation are written in Python, Java, and Scala.

\myparagraph{System Deployment}
Our proposed {VulnDS} is deployed in a loan management system of our collaborated bank. Figure~\ref{deployment} presents the system interface and main components, where the left part of Figure~\ref{deployment} presents the control and metric panel, including the risk statistics of each of the loan communities and control menus. The right part of Figure~\ref{deployment} displays the loan status monitoring screen. The node size indicates the delinquent probability of each company, which is dynamic and changes periodically according to the time window. Thus, risk managers could focus on risky and dominant companies.

\subsection{Case Study for Loan Default Prediction}

In this section, we conduct the case studies based on the deployed system.
We directly observe labels from real-world behavior and validate the prediction result with the tagged labels.
In the previous research, many machine learning based methods are developed for loan default prediction.
To further demonstrate the performance of proposed methods, we compare the proposed methods with some baseline approaches, which are designed for the default prediction task for real-world system. The baseline methods include Wide \cite{mcmahan2011follow}, Wide and Deep \cite{cheng2016wide}, CNN-max \cite{zheng2017joint}, GBDT \cite{ke2017lightgbm}, crDNN \cite{tan2018deep}, INDDP \cite{cheng2018prediction}, HGAR \cite{DBLP:conf/ijcai/ChengTMN019}.
We also employ betweenness~\cite{mumtaz2017identifying}, PageRank~\cite{langville2004deeper}, $k$-core~\cite{chen2021edge} and influence maximization (InfMax)~\cite{DBLP:conf/kdd/KempeKT03} methods as baselines.
We conduct the experiments over real-world dataset, i.e., Guarantee dataset, which spans 4 yeas, from 2012 to 2016.
As observed, most of the loans are repaid monthly.
Hence, we aggregate the behavior features within one-month time window and mark the delinquency loans as the target label for the month.
The records of 2012 are used as the training data. Then, we predict the defaults over the next three years.
For the baseline methods, the training data is used to train the prediction models.
For our methods, the training data is used to train the probabilities involved in the networks, which details are shown in our previous research \cite{cheng2018prediction}.

The results are shown in Table~\ref{tab:result}, where AUC (Area Under the Curve) for each year is reported. As we can see, GBDT and  Wide \& Deep outperform the Wide model, because of the increase of model capacity.
INDDP and HGAR are shown to be competitive across all the baselines.
Betweenness and PageRank are close to each other, which do not perform satisfactory. InfMax and K-core are much better which are still suboptimal.
BSR and \ours surpasses all the other approaches, which means
the graph structure and default diffusion properties are effective for  default prediction tasks.
BSR is slightly better than \ours, because it can offer tight theoretical guarantee.

\begin{table}[t]
  \center
  \caption{Results of default prediction}\label{tab:result}
  \renewcommand{\arraystretch}{1.1}
\setlength{\tabcolsep}{9pt}
  \begin{tabular}{lccc}
    \toprule
     & AUC(2014) & AUC(2015) & AUC(2016) \\
    \midrule
    Wide            & 0.75509  &  0.77751  &  0.78195     \\
    Wide \& Deep    & 0.76464  &  0.79825  &  0.81053     \\
    GBDT            & 0.77263  &  0.80627  &  0.81182     \\
    CNN-max         & 0.77645  &  0.80049  &  0.81492     \\
    crDNN           & 0.77429  &  0.79565  &  0.81054     \\
    INDDP           & 0.79015  &  0.80927  &  0.81588     \\
    HGAR            & 0.81310  &  0.80988  &  0.81875     \\
    Betweenness     & 0.60649  &  0.60577  &  0.60095     \\
  PageRank        & 0.61359  &  0.61643  &  0.62475     \\
    K-core          & 0.65551  &  0.66281  &  0.66816     \\
    InfMax          & 0.70255  &  0.70927  &  0.71132     \\ \midrule
    \ours           & 0.82367  &  0.82835  &  0.83709     \\
    BSR           & \textbf{0.82539}$^{**}$  &  \textbf{0.83004}$^{**}$  &  \textbf{0.83917}$^{**}$ \\
  \bottomrule
\end{tabular}
\vspace{-0pt}
\end{table}

\section{Related Work}
\label{subsec:related_work}


\myparagraph{Uncertain graph}
The uncertain (probabilistic) graph, where each node or edge may appear with a certain probability, has been widely used to model graphs with uncertainty in a wide spectrum of graph applications.
A large number of classical graph problems have been studied in the context of probabilistic graphs.
For instance, \cite{jin2011distance} investigates the distance-constraint reachability problem in probabilistic  graph. \cite{potamias2010k} introduces a framework to address the $k$ nearest neighbors (kNN) queries on probabilistic graphs.
\cite{khan2018conditional} investigate the reliability problem based on conditional probability. In \cite{ke2019depth}, authors provide a comprehensive comparison between different reliability algorithms.
The problem of vulnerable nodes detection has been investigated in the context of network reliability (e.g.,~\cite{articleli13,DBLP:conf/infocom/SenMBDC14,Cetinay2018}).
Nevertheless, their models are inherently different with ours, and the techniques cannot be trivially applied.
The problem investigated in this paper is similar to the study of node influence under the independent cascade (IC) model~\cite{DBLP:conf/kdd/KempeKT03,li2018influence} in the sense that the influence of a node can be modeled by possible world semantics.
Although a large body of works (e.g.,~\cite{DBLP:conf/kdd/KempeKT03,DBLP:journals/tkde/WangZZLC17,DBLP:conf/sigmod/TangSX15,wang2016efficient}) have been developed for the problem of influence maximization under the IC model, their proposed techniques cannot be applied to our problem due to the inherent difference between the two problems. Firstly, the nodes in IC model do not carry any probability. Secondly, their focus is to select $k$ nodes such that the spread of influence is maximized. While we aim to find $k$ nodes with largest default probabilities.

\myparagraph{Credit evaluation}
Consumer credit risk evaluation is often technically addressed in a data-driven fashion and has been extensively investigated \cite{baesens2003using, hand1997statistical}. Since the seminal “Partial Credit” model \cite{masters1982rasch}, numerous statistical approaches have been introduced for credit scoring: logistic regression, k-NN, neural network, and support vector machine.
More recently, \cite{baesens2003using} presents an in-depth analysis on interpreting the learned knowledge embedded in neural networks by using explanatory rules, and discusses how to visualize these rules. Researchers combine debt-to-income ratio with consumer banking transactions and use a linear regression model with time-windowed dataset to predict the default rates in a short future. They claim an 85\% default prediction accuracy and can save costs of between 6\% and 25\% \cite{khandani2010consumer}.

\myparagraph{Diffusion in finance}
The relationship between network structure and financial system risk has been carefully studied and several insights have been drawn. Network structure has little impact on system welfare, but it is important in determining systemic risk and welfare in short-term debt \cite{allen2010financial}.
Network theory attracts more attention after the 2008 global financial crisis. The crisis brought about by Lehman Brothers infects connected corporations, which is similar to the 2002 Severe Acute Respiratory Syndrome (SARS) epidemic.  Both of them start from small damages, but hit a networked society and cause serious events \cite{bougheas2015complex,cheng2020contagious}.
The dynamic network produced by bank overnight funds loans may be an alert of the crisis \cite{van2013using}. Moreover, research that aims to understand individual behavior and interactions in the social network has also attracted extensive attention \cite{onnela2006complex, qiu2016lifecycle}. Although preliminary efforts have been made using network theory to understand fundamental problems in financial systems \cite{chow2008social}, there is little work on system risk analysis in networked-guarantee loans \cite{meng2017netrating}.

\section{Conclusion}
\label{sec:conclusion}

In this paper, we investigate the problem of top-$k$ vulnerable nodes detection over uncertain graphs, which is very important for risk management in real-world applications. We formally define the problem by considering both self-risk probability of the nodes and the prorogation probability of defaults among graph nodes.
Due to the hardness of the problem, a sampling based method is developed with tight theoretical guarantee. To scale for large networks,
effective pruning techniques and advanced sampling method are proposed with rigorous theoretical analysis. To further accelerate the search, a bottom-$k$ based approach is presented.
We conduct extensive experiments over real-world datasets to demonstrate the efficiency and effectiveness of proposed techniques.
Moreover, the proposed techniques are integrated into our loan risk control system, which is deployed in real financial environment. Through case studies, we further verify the advantages of proposed model.

\bibliographystyle{unsrt}
\bibliography{references}

\begin{thebibliography}{10}

\bibitem{DBLP:conf/icde/LiDWMQY19}
Rong{-}Hua Li, Qiangqiang Dai, Guoren Wang, Zhong Ming, Lu~Qin, and Jeffrey~Xu
  Yu.
\newblock Improved algorithms for maximal clique search in uncertain networks.
\newblock In {\em {ICDE}}, 2019.

\bibitem{DBLP:journals/tkde/ZhangLZZZ16}
Wenjie Zhang, Xuemin Lin, Ying Zhang, Ke~Zhu, and Gaoping Zhu.
\newblock Efficient probabilistic supergraph search.
\newblock {\em {TKDE}}, 28(4), 2016.

\bibitem{DBLP:journals/pvldb/BoldiBGT12}
Paolo Boldi, Francesco Bonchi, Aristides Gionis, and Tamir Tassa.
\newblock Injecting uncertainty in graphs for identity obfuscation.
\newblock {\em {PVLDB}}, 5(11), 2012.

\bibitem{potamias2010k}
Michalis Potamias, Francesco Bonchi, Aristides Gionis, and George Kollios.
\newblock K-nearest neighbors in uncertain graphs.
\newblock {\em PVLDB}, 3(1-2):997--1008, 2010.

\bibitem{ke2019depth}
Xiangyu Ke, Arijit Khan, and Leroy Lim~Hong Quan.
\newblock An in-depth comparison of st reliability algorithms over uncertain
  graphs.
\newblock {\em {PVLDB}}, 12(8):864--876, 2019.

\bibitem{li2018influence}
Yuchen Li, Ju~Fan, Yanhao Wang, and Kian-Lee Tan.
\newblock Influence maximization on social graphs: A survey.
\newblock {\em {TKDE}}, 30(10), 2018.

\bibitem{cheng2020delinquent}
Dawei Cheng, Zhibin Niu, and Liqing Zhang.
\newblock Delinquent events prediction in temporal networked-guarantee loans.
\newblock {\em IEEE Transactions on Neural Networks and Learning Systems},
  2020.

\bibitem{jian2012determinants}
Ming Jian and Ming Xu.
\newblock Determinants of the guarantee circles: The case of chinese listed
  firms.
\newblock {\em Pacific-Basin Finance Journal}, 20(1):78--100, 2012.

\bibitem{mcmahon2014loan}
Dinny Mcmahon.
\newblock Loan ‘guarantee chains’ in china prove flimsy.
\newblock {\em The Wall Street Journal}, 27, 2014.

\bibitem{DBLP:conf/ijcai/ChengTMN019}
Dawei Cheng, Yi~Tu, Zhen{-}Wei Ma, Zhibin Niu, and Liqing Zhang.
\newblock Risk assessment for networked-guarantee loans using high-order graph
  attention representation.
\newblock In {\em IJCAI}, pages 5822--5828, 2019.

\bibitem{DBLP:conf/sigmod/AbiteboulKG87}
Serge Abiteboul, Paris~C. Kanellakis, and G{\"{o}}sta Grahne.
\newblock On the representation and querying of sets of possible worlds.
\newblock In {\em SIGMOD}, 1987.

\bibitem{niu2020iconviz}
Zhibin Niu, Runlin Li, Junqi Wu, Dawei Cheng, and Jiawan Zhang.
\newblock iconviz: Interactive visual exploration of the default contagion risk
  of networked-guarantee loans.
\newblock In {\em VAST}, pages 84--94, 2020.

\bibitem{DBLP:conf/ijcai/ChengWZ020}
Dawei Cheng, Xiaoyang Wang, Ying Zhang, and Liqing Zhang.
\newblock Risk guarantee prediction in networked-loans.
\newblock In Christian Bessiere, editor, {\em {IJCAI}}, pages 4483--4489, 2020.

\bibitem{DBLP:conf/kdd/KempeKT03}
David Kempe, Jon~M. Kleinberg, and {\'{E}}va Tardos.
\newblock Maximizing the spread of influence through a social network.
\newblock In {\em {SIGKDD}}, 2003.

\bibitem{cheng2018prediction}
Dawei Cheng, Zhibin Niu, Yi~Tu, and Liqing Zhang.
\newblock Prediction defaults for networked-guarantee loans.
\newblock In {\em 24th International Conference on Pattern Recognition}, pages
  361--366, 2018.

\bibitem{khan2014fast}
Arijit Khan, Francesco Bonchi, Aristides Gionis, and Francesco Gullo.
\newblock Fast reliability search in uncertain graphs.
\newblock In {\em EDBT}, pages 535--546, 2014.

\bibitem{cohen2007summarizing}
Edith Cohen and Haim Kaplan.
\newblock Summarizing data using bottom-k sketches.
\newblock In {\em Proceedings of the twenty-sixth annual ACM symposium on
  Principles of distributed computing}, pages 225--234, 2007.

\bibitem{Borgs:2014}
Christian Borgs, Michael Brautbar, Jennifer Chayes, and Brendan Lucier.
\newblock Maximizing social influence in nearly optimal time.
\newblock In {\em {SODA}}, pages 946--957, 2014.

\bibitem{anand2015filling}
Kartik Anand, Ben Craig, and Goetz Von~Peter.
\newblock Filling in the blanks: Network structure and interbank contagion.
\newblock {\em Quantitative Finance}, 15(4):625--636, 2015.

\bibitem{fu2016credit}
Kang Fu, Dawei Cheng, Yi~Tu, and Liqing Zhang.
\newblock Credit card fraud detection using convolutional neural networks.
\newblock In {\em International Conference on Neural Information Processing},
  pages 483--490, 2016.

\bibitem{engelmann2011basel}
Bernd Engelmann and Robert Rauhmeier.
\newblock {\em The basel II risk parameters: estimation, validation, stress
  testing-with applications to loan risk management}.
\newblock Springer Science \& Business Media, 2011.

\bibitem{thu2017transforming}
Ei~Ei Thu and Nwe Nwe.
\newblock Transforming model oriented program into android source code based on
  drools rule engine.
\newblock {\em Journal of Computer and Communications}, 5(03):49, 2017.

\bibitem{patel2015sales}
Ming Xu, Xiaowei Xu, Jian Xu, Yizhi Ren, Haiping Zhang, and Ning Zheng.
\newblock A forensic analysis method for redis database based on rdb and aof
  file.
\newblock {\em Journal of Computers}, 9(11):2538--2544, 2014.

\bibitem{jacomy2014forceatlas2}
M~Jacomy, S~Heymann, T~Venturini, and M~Bastian.
\newblock Forceatlas2, a figure layout algorithm for handy network
  visualization.
\newblock {\em Sciences Po}, 44, 2012.

\bibitem{mcmahan2011follow}
HB~McMahan.
\newblock Follow-the-regular ized-leader and mil-ror descent: Equivalence
  theorems and 11 regularization.
\newblock {\em Journal of Machine Learning Research Proceedings Trade},
  15:525--533, 2011.

\bibitem{cheng2016wide}
Heng-Tze Cheng, Levent Koc, Jeremiah Harmsen, Tal Shaked, Tushar Chandra,
  Hrishi Aradhye, Glen Anderson, Greg Corrado, Wei Chai, Mustafa Ispir, et~al.
\newblock Wide \& deep learning for recommender systems.
\newblock In {\em Proceedings of the 1st Workshop on Deep Learning for
  Recommender Systems}, pages 7--10, 2016.

\bibitem{zheng2017joint}
Lei Zheng, Vahid Noroozi, and Philip~S Yu.
\newblock Joint deep modeling of users and items using reviews for
  recommendation.
\newblock In {\em ACM International Conference on Web Search and Data Mining},
  2017.

\bibitem{ke2017lightgbm}
Guolin Ke, Qi~Meng, Thomas Finley, Taifeng Wang, Wei Chen, Weidong Ma, Qiwei
  Ye, and Tie-Yan Liu.
\newblock Lightgbm: A highly efficient gradient boosting decision tree.
\newblock In {\em Advances in Neural Information Processing Systems}, pages
  3146--3154, 2017.

\bibitem{tan2018deep}
Fei Tan, Xiurui Hou, Jie Zhang, Zhi Wei, and Zhenyu Yan.
\newblock A deep learning approach to competing risks representation in
  peer-to-peer lending.
\newblock {\em IEEE transactions on neural networks and learning systems},
  2018.

\bibitem{mumtaz2017identifying}
Sara Mumtaz and Xiaoyang Wang.
\newblock Identifying top-k influential nodes in networks.
\newblock In {\em {CIKM}}, 2017.

\bibitem{langville2004deeper}
Amy~N Langville and Carl~D Meyer.
\newblock Deeper inside pagerank.
\newblock {\em Internet Mathematics}, 1(3):335--380, 2004.

\bibitem{chen2021edge}
Chen Chen, Qiuyu Zhu, Renjie Sun, Xiaoyang Wang, and Yanping Wu.
\newblock Edge manipulation approaches for k-core minimization: Metrics and
  analytics.
\newblock {\em TKDE}, 2021.

\bibitem{jin2011distance}
Ruoming Jin, Lin Liu, Bolin Ding, and Haixun Wang.
\newblock Distance-constraint reachability computation in uncertain graphs.
\newblock {\em Proceedings of the VLDB Endowment}, 4(9):551--562, 2011.

\bibitem{khan2018conditional}
Arijit Khan, Francesco Bonchi, Francesco Gullo, and Andreas Nufer.
\newblock Conditional reliability in uncertain graphs.
\newblock {\em {TKDE}}, 30(11), 2018.

\bibitem{articleli13}
Shudong Li, Lixiang Li, Xinran Liu, and Yixian Yang.
\newblock Identifying vulnerable nodes of complex networks in cascading
  failures induced by node-based attacks.
\newblock {\em Mathematical Problems in Engineering}, 2013:938398, 09 2013.

\bibitem{DBLP:conf/infocom/SenMBDC14}
Arunabha Sen, Anisha Mazumder, Joydeep Banerjee, Arun Das, and Randy Compton.
\newblock Identification of {K} most vulnerable nodes in multi-layered network
  using a new model of interdependency.
\newblock In {\em {IEEE} {INFOCOM} Workshops}, pages 831--836, 2014.

\bibitem{Cetinay2018}
Hale Cetinay, Karel Devriendt, and Piet Van~Mieghem.
\newblock Nodal vulnerability to targeted attacks in power grids.
\newblock {\em Applied Network Science}, 3(1):34, 2018.

\bibitem{DBLP:journals/tkde/WangZZLC17}
Xiaoyang Wang, Ying Zhang, Wenjie Zhang, Xuemin Lin, and Chen Chen.
\newblock Bring order into the samples: {A} novel scalable method for influence
  maximization.
\newblock {\em TKDE}, 29(2):243--256, 2017.

\bibitem{DBLP:conf/sigmod/TangSX15}
Youze Tang, Yanchen Shi, and Xiaokui Xiao.
\newblock Influence maximization in near-linear time: {A} martingale approach.
\newblock In {\em {SIGMOD}}, 2015.

\bibitem{wang2016efficient}
Xiaoyang Wang, Ying Zhang, Wenjie Zhang, and Xuemin Lin.
\newblock Efficient distance-aware influence maximization in geo-social
  networks.
\newblock {\em TKDE}, 29(3):599--612, 2016.

\bibitem{baesens2003using}
Bart Baesens, Rudy Setiono, Christophe Mues, and Jan Vanthienen.
\newblock Using neural network rule extraction and decision tables for
  credit-risk evaluation.
\newblock {\em Management science}, 49(3):312--329, 2003.

\bibitem{hand1997statistical}
David~J Hand and William~E Henley.
\newblock Statistical classification methods in consumer credit scoring: a
  review.
\newblock {\em Journal of the Royal Statistical Society: Series A (Statistics
  in Society)}, 160(3):523--541, 1997.

\bibitem{masters1982rasch}
Geoff~N Masters.
\newblock A rasch model for partial credit scoring.
\newblock {\em Psychometrika}, 47(2):149--174, 1982.

\bibitem{khandani2010consumer}
Amir~E Khandani, Adlar~J Kim, and Andrew~W Lo.
\newblock Consumer credit-risk models via machine-learning algorithms.
\newblock {\em Journal of Banking \& Finance}, 34(11):2767--2787, 2010.

\bibitem{allen2010financial}
Franklin Allen, Ana Babus, and Elena Carletti.
\newblock Financial connections and systemic risk.
\newblock Technical report, National Bureau of Economic Research, 2010.

\bibitem{bougheas2015complex}
Spiros Bougheas and Alan Kirman.
\newblock Complex financial networks and systemic risk: A review.
\newblock In {\em Complexity and Geographical Economics}, pages 115--139. 2015.

\bibitem{cheng2020contagious}
Dawei Cheng, Zhibin Niu, and Yiyi Zhang.
\newblock Contagious chain risk rating for networked-guarantee loans.
\newblock In {\em KDD}, pages 2715--2723, 2020.

\bibitem{van2013using}
V{\'e}ronique Van~Vlasselaer, Jan Meskens, Dries Van~Dromme, and Bart Baesens.
\newblock Using social network knowledge for detecting spider constructions in
  social security fraud.
\newblock In {\em ASONAM}, pages 813--820, 2013.

\bibitem{onnela2006complex}
Jukka-Pekka Onnela et~al.
\newblock {\em Complex networks in the study of financial and social systems}.
\newblock Helsinki University of Technology, 2006.

\bibitem{qiu2016lifecycle}
Jiezhong Qiu, Yixuan Li, Jie Tang, Zheng Lu, Hao Ye, Bo~Chen, Qiang Yang, and
  John~E Hopcroft.
\newblock The lifecycle and cascade of wechat social messaging groups.
\newblock In {\em WWW}, pages 311--320, 2016.

\bibitem{chow2008social}
Wing~S Chow and Lai~Sheung Chan.
\newblock Social network, social trust and shared goals in organizational
  knowledge sharing.
\newblock {\em Information \& Management}, 45(7):458--465, 2008.

\bibitem{meng2017netrating}
Xiangfeng Meng, Yunhai Tong, Xinhai Liu, Yiren Chen, and Shaohua Tan.
\newblock Netrating: Credit risk evaluation for loan guarantee chain in china.
\newblock In {\em Pacific-Asia Workshop on Intelligence and Security
  Informatics}, pages 99--108, 2017.

\end{thebibliography}

\end{document}